\documentclass[12pt]{article}

\usepackage{amssymb,amsmath,amsfonts,mathtools}
\usepackage{geometry}
\usepackage{pdflscape}
\usepackage[normalem]{ulem}
\usepackage{graphicx}
\usepackage{caption}
\usepackage{setspace}
\usepackage{sectsty}
\usepackage{comment}
\usepackage{footmisc}
\usepackage{array}
\usepackage{booktabs}
\usepackage{natbib}
\usepackage{newtxtext,newtxmath}
\usepackage{enumitem}
\usepackage{float}
\usepackage{multirow}
\usepackage{threeparttable}
\usepackage{xcolor}
\usepackage{etoolbox}
\usepackage{hyperref}

\newcommand{\ignore}[1]{}

\AtBeginEnvironment{table}{\setstretch{1}}
\AtBeginEnvironment{figure}{\setstretch{1}}

\ifdefined\anonymous
\else
\fi
\def\doublespacing{\setstretch{2}}

\newtheorem{proposition}{Proposition}

\newtheorem{lemma}{Lemma}

\definecolor{wine}{RGB}{114,23,38}
\hypersetup{colorlinks=true,linkcolor=wine,citecolor=wine,urlcolor=wine,
  filecolor=wine}

\graphicspath{{figures/}}

\newcommand{\SampleStart}{July 1963}
\newcommand{\OOSStart}{August 1973}
\newcommand{\OOSEnd}{December 2023}
\newcommand{\NOOS}{605}

\newcommand{\NBurnIn}{121}
\newcommand{\TimingWindow}{120}
\newcommand{\NAnomaly}{212}
\newcommand{\NTiming}{30}
\newcommand{\NSignals}{242}
\newcommand{\NSignalsPolicy}{243}
\newcommand{\NCoefficientsPolicy}{1,215}
\newcommand{\RiskAversion}{5}
\newcommand{\TimingCap}{0.6}
\newcommand{\TimingLeverage}{1.5}
\newcommand{\TimingPcaRank}{10}

\newcommand{\TimingOmega}{5.80}
\newcommand{\TimingActive}{7.05}
\newcommand{\TimingStatic}{4.97}
\newcommand{\TimingComponent}{2.08}
\newcommand{\TimingT}{3.00}
\newcommand{\TimingShare}{29.5}
\newcommand{\TimingGrossSR}{1.242}
\newcommand{\TimingMapTiming}{1.50}
\newcommand{\TimingMapT}{1.53}
\newcommand{\TimingMarketSR}{0.739}
\newcommand{\TimingConstSR}{1.193}
\newcommand{\TimingCmvTiming}{0.98}
\newcommand{\TimingCmvT}{1.15}
\newcommand{\TimingLassoTiming}{0.82}
\newcommand{\TimingLassoT}{1.00}
\newcommand{\TimingMinvarTiming}{0.11}
\newcommand{\TimingInvvarTiming}{\ensuremath{-}0.12}
\newcommand{\TimingPConst}{0.953}
\newcommand{\TimingPMap}{0.953}
\newcommand{\TimingSeventies}{\ensuremath{-}0.14}

\newcommand{\TimingEighties}{0.79}

\newcommand{\TimingNineties}{2.93}

\newcommand{\TimingThousands}{\ensuremath{-}0.06}

\newcommand{\TimingTens}{0.72}

\newcommand{\TimingTwenties}{1.89}

\newcommand{\TimingExNineties}{1.73}
\newcommand{\TimingExNinetiesT}{2.13}
\newcommand{\TimingExNinetiesMonths}{485}
\newcommand{\TimingRobustLow}{1.47}
\newcommand{\TimingRobustHigh}{3.02}
\newcommand{\TimingRobustTLow}{2.18}
\newcommand{\TimingRobustTHigh}{4.04}

\newcommand{\RevScaleMatched}{0.375}
\newcommand{\RevScaledTiming}{2.04}
\newcommand{\RevScaledT}{2.65}

\newcommand{\RevScaledDiffT}{0.25}

\newcommand{\RevMonthlyTiming}{1.55}
\newcommand{\RevMonthlyT}{2.19}

\newcommand{\RevBagTiming}{1.69}
\newcommand{\RevBagT}{1.95}

\newcommand{\RevBagAtLimit}{86.6}

\newcommand{\RevRidgeHalfAtLimit}{62.5}

\newcommand{\RevRidgeTimingLow}{1.47}
\newcommand{\RevRidgeTimingHigh}{1.69}
\newcommand{\RevPostSD}{1.00}
\newcommand{\RevBootSD}{0.42}
\newcommand{\RevAtLimitBPPP}{52.6}
\newcommand{\RevAtLimitMap}{96.4}
\newcommand{\RevContraction}{0.45}
\newcommand{\RevContractionLow}{0.32}
\newcommand{\RevContractionHigh}{0.54}
\newcommand{\RevContractionCorr}{0.983}

\newcommand{\RevPeakScaleTiming}{0.400}
\newcommand{\RevPeakTiming}{2.05}
\newcommand{\RevPeakScaleSR}{0.425}
\newcommand{\RevPeakSR}{1.245}
\newcommand{\RevMandateBTLow}{2.19}
\newcommand{\RevMandateBTHigh}{3.06}
\newcommand{\RevMandateMTLow}{1.39}
\newcommand{\RevMandateMTHigh}{2.18}

\newcommand{\RevMandateTimLow}{1.38}
\newcommand{\RevMandateTimHigh}{3.70}

\newcommand{\RevMandateNetLow}{0.963}
\newcommand{\RevMandateNetHigh}{1.107}
\newcommand{\RevRecTiming}{1.90}
\newcommand{\RevRecT}{2.64}

\newcommand{\RevRecModal}{72}

\newcommand{\RevRecCostT}{1.93}
\newcommand{\RevRecCostNetSR}{1.121}

\newcommand{\RevRecScaledTiming}{2.00}
\newcommand{\RevRecScaledT}{2.46}

\newcommand{\RevCostEightyNet}{1.118}
\newcommand{\RevCostEightyTurn}{1.51}
\newcommand{\RevCostEightyTiming}{1.42}
\newcommand{\RevCostEightyT}{2.25}
\newcommand{\RevCostEightyP}{0.79}
\newcommand{\RevCostTLow}{2.03}

\newcommand{\RevRecCostP}{0.83}
\newcommand{\RevRecCostCEDiff}{\ensuremath{-}0.34}

\newcommand{\RevConstNetSR}{1.136}

\newcommand{\RevScaleTwoNet}{1.122}
\newcommand{\RevScaleTwoT}{1.78}
\newcommand{\RevScaleTwoTiming}{1.22}
\newcommand{\RevScaleTwoTurn}{1.09}
\newcommand{\RevExcessGrossSR}{0.960}

\newcommand{\RevExcessConstGross}{0.925}

\newcommand{\RevDropAnomaly}{0.87}
\newcommand{\RevDropAnomalyT}{1.13}
\newcommand{\RevFamilyMax}{0.05}
\newcommand{\RevPub}{1.22}
\newcommand{\RevPubT}{1.79}

\newcommand{\RevPubPost}{1.75}

\newcommand{\RevPubTwo}{1.17}
\newcommand{\RevPubTwoT}{1.75}

\newcommand{\RevPost}{1.82}

\newcommand{\RevPreBPPP}{1.48}
\newcommand{\RevPreMap}{2.44}
\newcommand{\RevExactDates}{605}

\newcommand{\RevExactRhatMax}{1.005}
\newcommand{\RevExactMcseMax}{0.0037}
\newcommand{\RevExactMcseMedian}{0.0014}
\newcommand{\RevExactMinRank}{113}
\newcommand{\RevExactAccept}{95}
\newcommand{\RevExactSdRatio}{0.996}
\newcommand{\RevExactMeanShift}{0.026}
\newcommand{\RevExactMaxDiff}{0.070}
\newcommand{\RevExactMeanDiff}{0.0086}
\newcommand{\RevExactTiming}{2.12}
\newcommand{\RevExactT}{2.95}

\newcommand{\RevPcaTen}{4.42}
\newcommand{\RevPcaTenT}{4.51}

\newcommand{\RevPcaTenExpl}{53}
\newcommand{\RevPcaTurnLow}{10.6}
\newcommand{\RevPcaTurnHigh}{15.4}
\newcommand{\RevPcaNetLow}{0.56}
\newcommand{\RevPcaNetHigh}{0.85}

\newcommand{\RevTurnover}{5.01}

\newcommand{\RevBreakEighty}{51}
\newcommand{\RevBreakRec}{51}
\newcommand{\RevRecAlpha}{0.83}
\newcommand{\RevRecAlphaT}{0.76}

\newcommand{\RevBlendW}{40}
\newcommand{\RevBlendSR}{1.147}

\begin{document}

\title{Bayesian Parametric Portfolio Policies}
\ifdefined\anonymous
  \author{}
\else
  \author{Miguel C. Herculano\thanks{Adam Smith Business School, University of
  Glasgow. E: \href{mailto:miguel.herculano@glasgow.ac.uk}{miguel.herculano@glasgow.ac.uk}.
  W: \href{https://mcherculano.github.io}{mcherculano.github.io}. I thank Nadja
  Klein, Giovanni Roman\`o and Eva L\"utkebohmert, two anonymous referees and seminar participants at the
  Karlsruhe Institute of Technology and the University of Freiburg, for helpful
  comments.}}
\fi
\date{This Version: \today}

\ifdefined\anonymous
  \setstretch{1.5}
  \maketitle
  \thispagestyle{plain}
  \begin{abstract}
  \noindent Parametric portfolio policies map predictors into portfolio weights
without modelling returns, making them a natural tool for factor timing. With many
predictors, if the policy has more coefficients than observations the data cannot pin it down. We develop Bayesian Parametric Portfolio Policies
(BPPP), which discipline the policy with a prior, update it with realised
utility, and average the weights over the posterior. We contribute
to the factor-timing literature by estimating timing as a portfolio rule rather
than from return forecasts. With 242 predictors and five factors, BPPP earns
\TimingComponent{} percentage points a year from timing, mostly from aggregate
anomaly returns. Under a leverage limit, averaging sizes each position by its
uncertainty and comes close, without a separate scale choice, to the size that
maximises timing with hindsight. A turnover penalty preserves significant gross
timing at lower turnover.

  \end{abstract}
  \newpage
\else
  \begin{titlepage}
  \maketitle
  \vspace{1.5em}
  \begin{abstract}
  \noindent Parametric portfolio policies map predictors into portfolio weights
without modelling returns, making them a natural tool for factor timing. With many
predictors, if the policy has more coefficients than observations the data cannot pin it down. We develop Bayesian Parametric Portfolio Policies
(BPPP), which discipline the policy with a prior, update it with realised
utility, and average the weights over the posterior. We contribute
to the factor-timing literature by estimating timing as a portfolio rule rather
than from return forecasts. With 242 predictors and five factors, BPPP earns
\TimingComponent{} percentage points a year from timing, mostly from aggregate
anomaly returns. Under a leverage limit, averaging sizes each position by its
uncertainty and comes close, without a separate scale choice, to the size that
maximises timing with hindsight. A turnover penalty preserves significant gross
timing at lower turnover.

  \end{abstract}
  \setcounter{page}{0}
  \thispagestyle{empty}
  \end{titlepage}
  \pagebreak
  \newpage
\fi
\ifdefined\anonymous
  \setstretch{1.5}
\else
  \doublespacing
\fi

\section{Introduction}
\label{sec:intro}
A central problem in asset pricing is to translate predictive information into
portfolio decisions in a way that is both economically meaningful and empirically
robust. Decades of research document that expected returns vary predictably with
firm characteristics and macroeconomic signals \citep{ChenZimmermann2022, Gu2020,
Haddad2020}, yet exploiting such predictability in real time remains challenging.
Classic mean--variance approaches require modelling the joint conditional
distribution of returns, a task that quickly becomes infeasible in high
dimensions and is known to produce unstable weights in finite samples
\citep{DeMiguel2009, Green2017}. Parametric portfolio policies (PPP) offer an
appealing alternative \citep{BrandtSantaClara2006, Brandt2009}. Rather than
modelling returns, they parameterise portfolio weights directly as functions of
observable predictors and estimate the policy parameters by maximising average
realised utility. Estimation error then enters directly into the weights, which is
where the investor bears it.

Factor timing is a natural application for PPP. Factor premia vary over time, and
a large literature proposes variables that predict them \citep{Haddad2020,
Ilmanen2021, Neuhierl2023}. An investor who wants to use this evidence must turn
many noisy predictors into a few factor weights, which is the problem parametric
policies are built to solve. Two forms of the policy are in use.
\citet{Brandt2009} apply one coefficient vector to firm characteristics that are
standardised across stocks at each date. The rule is parsimonious and scales to
thousands of stocks, but it cannot use a predictor that takes the same value for
every asset, since standardisation removes the level of every characteristic and a
common coefficient can only move all assets together. The conditional policies of
\citet{BrandtSantaClara2006} give each asset its own coefficients on the
predictors, so the weight on each factor can respond to an aggregate state in its
own way. This is the form a factor timer needs, and it is the form we use. This form, however, is bound by the curse of dimensionality. With five factors and \NSignals{} aggregate
predictors, the policy has \NCoefficientsPolicy{} coefficients, which we
re-estimate every month from a rolling window of \TimingWindow{} months. The
empirical application of \citet{BrandtSantaClara2006} has ten coefficients and 671
months, and their standard errors, which follow \citet{BrittenJones1999}, require
more observations than coefficients. In our setting the sample criterion has no
unique maximum, and the data alone cannot determine the policy. Moreover, classical PPP treats the estimated coefficients as if they were known. Once a
point estimate is obtained, portfolio decisions are taken as if it were the true
policy. 

With far more coefficients than observations, this neglect of policy
uncertainty is consequential. The investor knows that the policy is uncertain, yet
behaves as if it is not. Bayesian learning addresses the curse of dimensionality in PPP and resolves the policy indeterminacy by introducing a prior over the policy coefficients that carries uncertainty into the portfolio weights. The prior shrinks each factor's
baseline weight towards the market and the timing coefficients towards zero, as
\citet{BrandtSantaClara2006} suggest. We combine it with realised
utility through the generalised posterior of \citet{Bissiri2016}, which updates
beliefs about the coefficients of a decision rule without a likelihood for
returns. Each factor weight is then averaged over the posterior after it is
constrained to lie within the position limit. We call the resulting policy a
Bayesian Parametric Portfolio Policy (BPPP). The posterior mode, which we show to be equivalent to a
ridge-penalised PPP and which we call RPPP, serves as the natural point-estimate
counterpart. Both use the same smooth, strongly convex training criterion, so the mode is
unique and its objective error can be bounded, and the averaged weights have a
closed form under the Gaussian approximation.

We use this policy to answer three questions that speak to the factor-timing literature. First,
\emph{are factors timed?} We split each rule's active return into a static part,
the average weight times the average factor return, and a timing part, the
covariance between the weight and the subsequent factor return \citep{Lo2008}.
BPPP earns \TimingComponent{} percentage points a year from timing, with a
Newey--West $t$-statistic of \TimingT, out of an active return of \TimingActive{}
points. The estimate is \RevRecTiming{} points when the prior is
chosen recursively, with no knowledge of later returns. Most of the timing comes
from aggregate anomaly signals, and the value factor accounts for the largest
share. Using each anomaly predictor only from the year in which it was published
lowers the estimate to \RevPub{} points ($t=\RevPubT$), but after 2000, when
nearly all predictors were public, timing is of similar size in every version,
about \RevPost{} points a year.

Second, \emph{what survives a leverage limit?} We hold every rule to the same
limits on individual factor positions and on gross factor exposure. RPPP is
subject to the gross-exposure limit in \RevAtLimitMap\% of months and must be
scaled down in proportion. Proportional scaling keeps the direction of each tilt
but compresses its variation over time, and this variation is precisely what
timing is made of. BPPP instead scales each factor weight back before the limit
binds, and scales it back more the less certain the posterior is about that
factor relative to the position limit. The local approximation keeps a median
fraction \RevContraction{} of the raw position implied by the posterior mode. This is
close to the fraction that would have maximised the timing of
RPPP had it been chosen with hindsight. BPPP therefore reaches a position size close to the retrospective optimum,
without choosing a separate exposure-scale parameter. The two natural alternatives, a
heavier ridge penalty and bootstrap aggregation, do not reach it.

Third, \emph{does factor timing survive trading costs?} Because the raw policy is
linear in the predictors, a fixed-coefficient turnover proxy in the estimation window
depends only on the data and the coefficients, so a quadratic turnover penalty in
the spirit of \citet{GarleanuPedersen2013} keeps the estimation problem convex.
The penalty cuts annual turnover from \RevTurnover{} to \RevCostEightyTurn{}
while the gross timing component remains significant ($t=\RevCostEightyT$). Net of
costs, the penalised policy is statistically indistinguishable from the best
static allocation, a recursively estimated optimal factor tilt that barely trades
(net Sharpe ratios of \RevCostEightyNet{} and \RevConstNetSR),
and its estimated net Sharpe ratio exceeds that allocation's when costs are below about
\RevBreakEighty\% of our schedule. Since the costs large institutions actually
pay are well below standard academic estimates \citep{FrazziniIsraelMoskowitz2012},
the value of timing depends on who implements it rather than on whether it
exists.

The paper contributes to the factor-timing literature in three ways.
\citet{Haddad2020} show that the leading principal components of anomaly returns
are predictable and that timing them generates large gains, and
\citet{Neuhierl2023} extend the analysis to more than 300 factors and 39
predictors, finding that timing improves almost every group of factors. Both build
timing portfolios from return forecasts. In contrast, we estimate factor timing directly
as a portfolio rule, in which the weights of five factors respond to 242
predictors, and we measure timing with a decomposition that applies on equal terms
to any rule. The resulting gain, about two percentage points a year, is of the same
order as the improvement \citet{Neuhierl2023} report for the median factor.
Second, we show how estimation uncertainty and leverage limits jointly
determine how much timing a policy retains, a question that studies of
unconstrained forecast-based portfolios do not address. Averaging over a posterior
sizes a timing policy without a separate exposure-scale search, and because much of the
value of accounting for estimation risk can be captured by simple shrinkage
\citep{TuZhou2011}, we show precisely how it differs from a ridge penalty and from
bootstrap aggregation. Third, we bring trading costs into the estimation itself
and measure the cost level at which timing pays, which speaks directly to the
scepticism of \citet{Arnott2019} and \citet{Ilmanen2021}.

The rest of the paper proceeds as follows. Section~\ref{sec:literature} reviews
the related literature. Section~\ref{sec:motivation} sets out the policy and
explains why factor timing requires its conditional form.
Section~\ref{sec:bppp} develops BPPP, its relation to ridge regression and the
turnover penalty. Section~\ref{sec:data} describes the data.
Section~\ref{sec:results} answers three questions on factor timing.
Section~\ref{sec:robustness} examines the sensitivity of the results and the
limits of the evidence, and Section~\ref{sec:conclusion} concludes.

\section{Related Literature}
\label{sec:literature}
Our policy is the conditional policy of \citet{BrandtSantaClara2006}, whose
managed portfolios go back to \citet{HansenRichard1987}, applied with a market
benchmark, position limits and a much larger set of predictors.
\citet{Brandt2009} and \citet{DeMiguel2020} study the characteristic form of the
policy in the cross-section of stocks, and \citet{DeMiguel2020} show how
trading costs change which characteristics matter.

A long literature treats parameter uncertainty as part of the portfolio problem
\citep{Kandel1996, Barberis2000, Pastor2000, Avramov2002, GarlappiUppalWang2007,
WachterWarusawitharana2009, Johannes2014, Pettenuzzo2014}, and
\citet{AvramovZhou2010} survey this literature. \citet{TuZhou2011} show that combining
estimated rules with naive diversification, a simple form of shrinkage, captures
much of the value of accounting for estimation risk. Because a likelihood for
returns is precisely what a parametric policy avoids, our posterior belongs to the
generalised Bayesian framework of \citet{Bissiri2016}, \citet{ChernozhukovHong2003}
and \citet{Lyddon2019}, and our diagonal learning-rate rule is motivated by
\citet{HolmesWalker2017} and \citet{Lyddon2019}.

Shrinkage and constraints are standard remedies for estimation error in large
portfolios \citep{JagannathanMa2003, LedoitWolf2004, KanZhou2007, DeMiguel2009}.
\citet{JagannathanMa2003} show that a constraint can act as shrinkage. We study
the reverse channel, namely how shrinkage determines what a leverage constraint
does to a timing strategy. Machine-learning approaches to the cross-section
\citep{Gu2020, Kozak2020, Bryzgalova2023, KellyMalamudZhou2024} rely on richer
priors than ours \citep{Carvalho2010, Piironen2017}.

On factor timing, \citet{Haddad2020} find large timing gains from predicting the
principal components of anomaly returns, and \citet{KozakNagelSantosh2018} and
\citet{Kozak2020} discipline such evidence with a stochastic discount factor (SDF).
\citet{Neuhierl2023} time more than 300 factors with 39 predictors and find that
timing improves the performance of almost every group of factors, most of all
profitability and value. \citet{Arnott2019} and \citet{Ilmanen2021} doubt that
timing survives implementation, and our cost results show when it does. The record
of return prediction out of sample \citep{WelchGoyal2008, GoyalWelch2024}, the
decay of anomalies after publication \citep{McLean2016} and the problem of multiple
testing \citep{HarveyLiu2016} motivate our recursive and publication-based checks.
Our treatment of trading costs follows \citet{GarleanuPedersen2013},
\citet{FrazziniIsraelMoskowitz2012} and \citet{NovyMarxVelikov2016}.

\section{Setup and Motivation}
\label{sec:motivation}
An investor holds $K$ assets with returns $R_{t+1}$ and observes a vector of $L$
predictors $z_t$ at date $t$. With CRRA utility $U$ and no hedging motive, the
investor chooses the weights $w_t$ to maximise $\mathbb E_t[U(1+w_t'R_{t+1})]$, with net returns and positive wealth. Doing
so directly requires the conditional distribution of $R_{t+1}$ given $z_t$, which
is difficult to estimate with many assets and predictors. A parametric portfolio
policy avoids this step by writing the weights as a function of the predictors,
\begin{equation}
	\label{policy}
	w_t(\theta) = w_b + \theta z_t,
\end{equation}
where $w_b$ is a benchmark and $\theta$ is a $K\times L$ matrix whose element
$\theta_{k\ell}$ is the response of the weight on asset $k$ to predictor $\ell$.
This is the conditional policy of \citet{BrandtSantaClara2006}. Each column of
$\theta z_t$ can be read as a managed portfolio that holds asset $k$ in proportion
to predictor $\ell$, and a fixed combination of these managed portfolios is a
dynamic strategy in the original assets. Given a sample
$\mathcal D_T=\{(R_{t+1},z_t)\}_{t=1}^T$, the coefficients solve
\begin{equation}
	\label{ppp_obj}
	\hat{\theta} = \arg\max_{\theta} \frac{1}{T} \sum_{t=1}^T
	U\!\left(1+w_t(\theta)'R_{t+1}\right).
\end{equation}
Equation~\eqref{ppp_obj} is a sample analogue of expected utility maximisation over
policies. It delivers state-dependent, implementable portfolio rules, it requires
no explicit return distribution, and it scales to large sets of predictors.

The formulation treats $\theta$ as an unknown but fixed object. Once a point
estimate $\hat\theta$ is obtained, portfolio decisions are taken as if this
estimate were known with certainty. This plug-in approach is standard in the
literature, but it neglects uncertainty about the policy. The omission is not consequential 
when $\theta$ is well determined and consequential when it is not. When $L$ is
large relative to $T$, $\hat\theta$ is estimated with substantial error, and
because weights are linear in $\theta$, small changes in $\hat\theta$ translate
directly into large and unstable changes in the weights from one period to the
next.

\subsection{Two forms of the policy}
\label{subsec:annihilation}

\citet{Brandt2009} propose a different form for the cross-section of stocks.
The weight on stock $i$ is
\begin{equation}
	\label{bsvpolicy}
	w_{i,t} = \bar{w}_{i,t} + \frac{1}{N_t}\,\vartheta' \hat{x}_{i,t},
\end{equation}
where $\hat{x}_{i,t}$ are characteristics of stock $i$, standardised across stocks
at date $t$, and the coefficient vector $\vartheta$ is common to all stocks. The
number of coefficients equals the number of characteristics, however many stocks
there are, which is what makes the rule practical for large cross-sections. The predictors used for factor timing are aggregate. A valuation spread or the
level of market volatility takes the same value for every factor at a given date,
and two properties of \eqref{bsvpolicy} rule out predictors of this kind.

First, standardising a characteristic across assets subtracts its
cross-sectional mean, so the level of a characteristic never reaches the
portfolio, shifting the characteristic by the same amount for every asset leaves
every standardised value, and hence every weight, unchanged. A predictor that takes
the same value for every asset has no cross-sectional variation at all, its
standardised value is undefined, and under the usual convention of setting it to
zero its coefficient has no effect on the portfolio. Second, a coefficient vector
that is common to all assets moves all positions together. If every row of
$\theta$ in \eqref{policy} were the same vector, an aggregate predictor could only
scale a fixed combination of the assets up or down, and could never shift the
weight from one asset to another.

Timing factors means changing the weight on one factor relative to the others as
conditions change, which is exactly what these two properties prevent. Neither is a
flaw of \eqref{bsvpolicy}, which was designed to select stocks on firm
characteristics. Pre-formed factor portfolios, however, have no firm
characteristics with which an aggregate predictor could interact, and for them the
conditional form \eqref{policy} is the one we use.\footnote{Aggregate predictors
could also enter a characteristic policy through interactions with asset-specific
characteristics or identities; the restriction here concerns \eqref{bsvpolicy}.}

\subsection{The factor policy and its size}
\label{subsec:recast}

In our setup, the investor holds the market with a weight of one and adds positions in $K-1$
zero-investment factors,
\begin{equation}
	\label{eq:factorpolicy}
	w_{m,t}=1,\qquad
	w_{k,t}(\theta)=\Pi_c(\theta_k'z_t),\quad k=1,\ldots,K-1,
\end{equation}
where $\Pi_c(u)=\min\{c,\max(-c,u)\}$ limits each factor position to $[-c,c]$ and
$\theta_k$ is the vector of coefficients for factor $k$. The benchmark is the
market, $w_b=(1,0,\ldots,0)'$, so $\theta=0$ holds no factor exposure and every
factor position must be earned by the policy. \citet{BrandtSantaClara2006}
suggest this benchmark-relative form for active managers. The portfolio return is
$r_{p,t+1}=r_{m,t+1}+w_t'F_{t+1}$, where $F_{t+1}$ collects the factor returns.
Because the factors are zero-investment portfolios, no budget constraint links
their positions to the market. The predictor vector $z_t$ contains a constant,
which gives each factor its own baseline position, and the \NSignals{} aggregate
predictors of Section~\ref{sec:data}. The investor also faces a limit $g$ on gross factor exposure, so the feasible
weights are
\begin{equation}
	\label{eq:feasible}
	\mathcal W=\Bigl\{(1,w')':\ |w_k|\le c\ \text{for all }k,\
	\textstyle\sum_k|w_k|\le g\Bigr\}.
\end{equation}
Since each factor has a long and a short leg, $\sum_k|w_k|$ measures factor
notional rather than exposure to individual securities. The conditional form has a cost. With five factors and $L=\NSignalsPolicy$
predictors, the policy has $\NCoefficientsPolicy$ coefficients, and we estimate
them each month from a rolling window of \TimingWindow{} observations. The
affine criterion \eqref{ppp_obj} is then flat along every combination of coefficients
that the window cannot distinguish, so it has no unique maximiser. The sampling
theory for the policy in \citet{BrandtSantaClara2006}, based on the regression
representation of \citet{BrittenJones1999}, requires more observations than
coefficients and is not available. The unpenalised PPP estimator \eqref{ppp_obj} hits the curse of dimensionality i.e.
it is not uniquely identified in our setting, and the data alone cannot determine the
policy. The next section develops the framework that overcomes this limitation.

\section{Bayesian Parametric Portfolio Policies}
\label{sec:bppp}
When $\theta$ is uncertain, we integrate over this uncertainty rather than
conditioning on a single estimated policy. We adopt a
Bayesian perspective and treat the policy parameter $\theta$ as a random variable.
Given a prior $p(\theta)$ and the loss from the affine PPP training objective,
the generalised posterior of \citet{Bissiri2016} is
\begin{equation} \label{eq:posterior}
	p_{\omega}(\theta \mid \mathcal{D}_T)
	\;\propto\;
	p(\theta) \exp\Bigl\{\omega\sum_{t=1}^T U\bigl(q_t(\theta)\bigr)\Bigr\}\,.
\end{equation}
Rather than updating beliefs through a conventional likelihood, which would
require a parametric model for returns, we update using the PPP objective itself.
\citet{Bissiri2016} show that this is the coherent way to update beliefs about a
quantity defined as the maximiser of an expected utility, which is what $\theta$
is, the coefficients of a decision rule rather than parameters of the return
distribution.

The learning rate $\omega$ sets the weight of the data relative to the prior.
Because utility is defined only up to a positive affine transformation, $\omega$
must be fixed before \eqref{eq:posterior} is a posterior. We draw on
\citet{HolmesWalker2017} and \citet{Lyddon2019}, who scale uncertainty by relating the spread
of the generalised posterior to the sampling variability of the underlying
estimator; \citet{SyringMartin2019} propose a related calibration based on
coverage. We use a diagonal approximation to this covariance-trace matching idea at each
date, and leave the details and its limitations in Appendix~\ref{sec:estimation}. The prior is Gaussian and centred at zero, with different scales for the two
kinds of coefficient
\begin{equation}\label{eq:hierprior}
\theta_{k\ell}\sim\mathcal N(0,\lambda_\ell^{-1}),\qquad
\lambda_0=\delta_{\mathrm s}^{-2},\qquad
\lambda_\ell=(L-1)\,\delta_{\tau}^{-2}\quad(\ell\ge1).
\end{equation}
The coefficient on the constant, $\theta_{k0}$, is the baseline raw position in factor
$k$, which a window of \TimingWindow{} months can estimate. The coefficients on
the predictors govern timing, and the window cannot estimate them individually.
The scale $\delta_{\mathrm s}$ controls how far baseline positions may move from
the market, and $\delta_\tau$ the reference scale of the timing part of a position.
Conditional on $z_t$, its prior variance is
$\delta_\tau^2\sum_{\ell\ge1}z_{t,\ell}^2/(L-1)$. This is the prior that
\citet[Section~III.B]{BrandtSantaClara2006} suggest, with intercepts shrunk towards
the benchmark and slopes towards zero. We use $\delta_{\mathrm s}=0.8$ and
$\delta_\tau=1$, and Section~\ref{subsec:timingrobust} chooses both recursively.

The negative logarithm of the posterior, over the \TimingWindow{} months in the
estimation window $\mathcal T_t$, is
\begin{equation}\label{eq:criterion}
f(\theta)=-\omega\sum_{s\in\mathcal T_t}U\{q_s(\theta)\}
+\tfrac12\,\theta'\Lambda\theta,
\qquad q_s(\theta)=1+r_{m,s+1}+\textstyle\sum_k\theta_k'z_s\,F_{k,s+1},
\end{equation}
where $U(q)=(q^{1-\gamma}-1)/(1-\gamma)$ with $\gamma=\RiskAversion$ and $\Lambda$
is the diagonal precision of \eqref{eq:hierprior}. The position limits are
applied to the portfolio weights, not inside $q_s$. This affine training loss
here replaces constrained portfolio utility. Wealth is then linear in
$\theta$, and $f$ is smooth and strongly convex on $\theta$. Its minimiser, the posterior mode,
is therefore unique, and every estimate comes with a computable bound on its
objective gap from the global optimum (see Lemma~\ref{lem:certified} in
Appendix~\ref{sec:proofs}).

\subsection{Posterior inference}
\label{subsec:inference}
Two portfolios follow from the posterior. The first uses its mode $\hat\theta_t$,
the minimiser of \eqref{eq:criterion}. Section~\ref{subsec:whybayes} shows that the
mode is a ridge-penalised PPP estimate, and we call the resulting policy RPPP. The
second, BPPP, averages the constrained weight over the posterior:
\begin{equation}
	\label{eq:bppp_weights}
	w_{k,t}^{\mathrm{BPPP}}
	= \mathbb{E}\bigl[\Pi_c(\theta_k'z_t)\bigr],
	\qquad
	w_{k,t}^{\mathrm{RPPP}}=\Pi_c(\hat\theta_{k,t}'z_t).
\end{equation}
In both cases the limit on gross exposure is imposed aposteriori, by scaling all
factor positions down in proportion when their absolute sum exceeds $g$.\footnote{The
policy averages individually bounded positions before gross scaling; in general
$S_g\{\mathbb E[\Pi_c(\theta z_t)]\}\ne\mathbb E[S_g\{\Pi_c(\theta z_t)\}]$.
This aggregation rule is not derived as a posterior expected-CRRA-utility optimum.}
The
unpenalised PPP estimator of Section~\ref{sec:motivation} does not identify a unique coefficient vector here, so
RPPP, not PPP, is the relevant point-estimate benchmark since it uses the same data,
the same prior and the same criterion as BPPP, and differs only in ignoring the
uncertainty that remains after the prior has been imposed.

Some remarks on posterior inference are in order. First, the expectation in
\eqref{eq:bppp_weights} is over a posterior for \NCoefficientsPolicy{}
coefficients, and it must be recomputed at each of \NOOS{} rebalancing dates. To ensure computation is feasible, we approximate the posterior by a normal distribution centred at the mode,
with covariance given by the curvature of \eqref{eq:criterion} there (a Laplace
approximation). Under this approximation each unconstrained weight
$\theta_k'z_t$ is normal, and the average of its constrained value has a closed
form (Appendix~\ref{app:average}), so BPPP involves no simulation. The approximation could matter, because the average is
taken through the nonlinear limit $\Pi_c$. We thus verify its quality by estimating the exact-target
posterior average by Monte Carlo at every date. This is possible because the criterion depends on
the coefficients only through the \TimingWindow{} wealth levels of the estimation
window, so the exact posterior can be sampled in \TimingWindow{} dimensions rather
than \NCoefficientsPolicy{} (Appendix~\ref{app:exact}). Section~\ref{subsec:timingrobust}
shows that the two give almost identical portfolios.

Second, position limits are the channel through which uncertainty about the
policy reaches the portfolio, and in our setting the effect is large. Because the
unconstrained weight $\theta_k'z_t$ is linear in $\theta$, constraining the
average weight would reproduce RPPP under the Gaussian approximation, whose
mean equals its mode. Averaging the constrained weight does
not, because draws beyond the limit contribute fixed boundary positions. The
local approximation has a median slope \RevContraction{} relative to the raw
position implied by the posterior mode. Before the leverage limit is applied, its average gross factor
exposure is \RevPreBPPP{} against \RevPreMap{} for RPPP, and after it BPPP is at the
limit in \RevAtLimitBPPP\% of months against \RevAtLimitMap\% for RPPP. BPPP and
RPPP are therefore different portfolios, and Section~\ref{subsec:mandate} shows
that the difference determines how much timing survives the limit. Since every
institutional mandate imposes limits of this kind, this is the relevant case
rather than a special one. The limit can bind because of demand for factor exposure as well as estimation
uncertainty. The constant allocation, which ignores the predictors and optimises
under both limits jointly, is also at the gross limit in every months. Proposition~\ref{prop:shrink} characterises the
difference for a single factor weight.

\begin{proposition}
	\label{prop:shrink}
	Fix a factor and a date, and suppose that under the posterior the
	unconstrained weight $X=\theta_k'z_t$ is normal with mean $\mu$, which is the
	weight implied by the posterior mode, and standard deviation $\sigma>0$. Let
	$\bar w=\mathbb E[\Pi_c(X)]$ be the BPPP weight before the limit on gross
	exposure, and $\Pi_c(\mu)$ the corresponding RPPP weight. Then:
	\begin{enumerate}[label=(\roman*)]
		\item \emph{Averaging shrinks the position towards the market.} If $\mu\neq0$,
		then $\bar w$ has the same sign as $\mu$ and
		$0<|\bar w|<|\Pi_c(\mu)|=\min(|\mu|,c)$. If $\mu=0$, then $\bar w=0$.
		\item \emph{More uncertainty means more shrinkage.} If $\mu\neq0$, $|\bar w|$
		falls strictly as $\sigma$ rises, and $\bar w\to\Pi_c(\mu)$ as $\sigma\to0$.
		\item \emph{Locally, shrinkage is set by uncertainty.} Let $\psi(\sigma)=2\Phi(c/\sigma)-1$,
		the probability that a normal variable with mean zero and standard deviation
		$\sigma$ lies inside $[-c,c]$. Then
		\[
		\bar w=\psi(\sigma)\,\mu+\varepsilon,\qquad
		|\varepsilon|\le 0.081\,|\mu|^3/\sigma^2 .
		\]
		\item \emph{With large uncertainty, the position reflects confidence in its
		sign.} If $\sigma$ grows with $\mu/\sigma$ held fixed, then
		$\bar w/c\to2\Phi(\mu/\sigma)-1$.
	\end{enumerate}
\end{proposition}

\noindent The proofs are in Appendix~\ref{sec:proofs}, and
Figure~\ref{fig:averaging} illustrates the result. Part~(i) says that averaging
within the limit always moves a factor position towards zero, that is, towards
the market, without ever reversing its sign. Part~(ii) says that it moves it
further the less certain the posterior is about that position. Parts~(iii) and (iv) give the size of the effect, and they link the proposition
to the evidence in Section~\ref{subsec:mandate}. Part~(iii) says that, for
positions that are not large relative to their uncertainty, BPPP holds approximately a fraction
$\psi(\sigma)$ of the position implied by the posterior mode, and that this
fraction depends only on how wide the limit is relative to the posterior standard
deviation of that position. The intuition is that moving the posterior of $X$ up
by a small amount moves the average constrained weight only through the draws
that lie inside the limit, because draws beyond it remain at $\pm c$; the share of
such draws at $\mu=0$ is $\psi(\sigma)$. The fraction is set by the posterior, position by
position, without reference to returns.

Part~(iv) covers the case in which uncertainty is large relative to the limit.
Almost every draw is then held at $+c$ or $-c$, so the BPPP weight becomes the
limit times the margin by which the posterior favours a long position over a
short one, $c\{2\Phi(\mu/\sigma)-1\}$. For small $\mu/\sigma$ this is about
$0.8\,c\,\mu/\sigma$. The position then depends on the ratio of the
posterior-mode weight to its standard deviation, not on the posterior-mode weight
alone. A large RPPP position that the data determine poorly receives a small BPPP
position, whereas RPPP holds it in full, up to the limit. This is the sense in
which BPPP sizes each factor position by its uncertainty.

\begin{figure}[t]\centering
\includegraphics[width=0.75\textwidth]{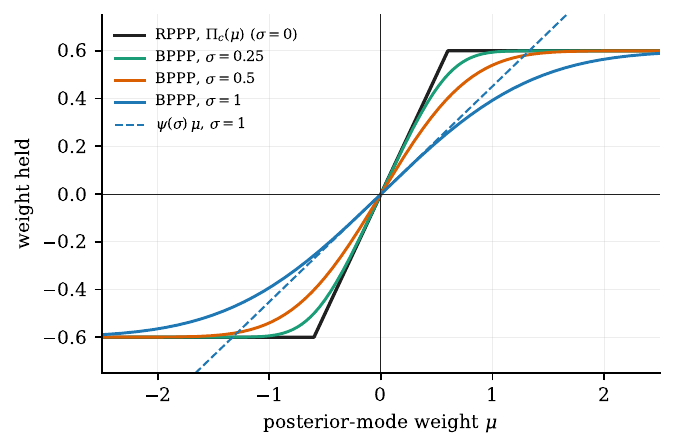}
\caption{Averaging within a position limit}\label{fig:averaging}
\begin{minipage}{\textwidth}\footnotesize
\textit{Notes:} The BPPP weight $\bar w=\mathbb E[\Pi_c(X)]$ with
$X\sim\mathcal N(\mu,\sigma^2)$, plotted against the posterior-mode weight $\mu$,
for a position limit $c=\TimingCap$. The solid black line is the RPPP weight
$\Pi_c(\mu)$, which corresponds to $\sigma=0$. The dashed line is the
approximation $\psi(\sigma)\mu$ of Proposition~\ref{prop:shrink}(iii) for
$\sigma=1$.
\end{minipage}
\end{figure}

If $\psi(\sigma)$ were the same for every factor and date, averaging would be
approximately equivalent to constant scaling locally, when $|\mu|<c$ and the
Taylor remainder is small. Outside that region, shrinkage also varies with $\mu$
even when uncertainty is constant. Averaging differs from that in two
respects. The posterior sets the amount of shrinkage, so no (arguably arbitrary) constant has to be
chosen, and the shrinkage varies across factors and dates with posterior
uncertainty. Section~\ref{subsec:mandate} asks whether a constant scaling of RPPP,
chosen with hindsight or recursively, can reach the timing of BPPP.

\subsection{Relation to ridge regression and bagging}
\label{subsec:whybayes}
The posterior mode is a penalised estimator, and it is useful to describe its properties to better understand our results, in particular the difference and added value between regularisation and accounting for parameter uncertainty within the Bayesian setting we outline. Dividing the criterion \eqref{eq:criterion} by $\omega$ does not change its
minimiser, and shows that the mode maximises summed realised utility minus the
quadratic penalty $\tfrac12\theta'(\Lambda/\omega)\theta$. The mode is therefore a
ridge-penalised PPP estimate, with a penalty that is heavier on the timing
coefficients than on the constants, which is why we call the policy based on it
RPPP. With average utility, the equivalent penalty precision is $\Lambda/(\omega T)$.
Any difference in returns between RPPP and another ridge policy reflects a
different penalty, not the Bayesian interpretation.

Before gross scaling, BPPP holds smaller absolute positions than RPPP, and there are two familiar ways to shrink a
point estimate that an investor might use instead. The first is a heavier ridge
penalty, which moves the mode itself. In our application it shrinks most the combinations of
coefficients that the window identifies least, which are the timing coefficients,
and shrinks least the average positions, which the data estimate well, so it
removes timing before it reduces the size of the positions. The second is bootstrap
aggregation, or bagging \citep{Breiman1996}. Bagged ridge re-estimates RPPP on many
resampled versions of the estimation window, constrains each resulting weight
within the position limit, and averages the constrained weights. It therefore
averages within the limit, as BPPP does, but over the sampling distribution of the
ridge estimator rather than over the posterior.\footnote{We draw 100 circular
block-bootstrap samples of the estimation window with blocks of 12 months,
re-estimate the mode on each with the same prior and learning rate, constrain each
weight to $[-c,c]$, average, and then apply the gross-exposure limit as for every
other rule.} The two
distributions differ in an important way. A ridge estimate computed on a
resampled window lies in the span of that window's data, so the bootstrap
distribution has no variance in the directions the window cannot identify. With
\NCoefficientsPolicy{} coefficients and \TimingWindow{} observations these are
most directions. The posterior assigns them their prior variance, and much of the
uncertainty about a factor weight comes from them. Section~\ref{subsec:mandate}
shows that these differences determine how much timing survives the leverage
limit.

\subsection{A turnover penalty}
\label{subsec:costaware}
The criterion \eqref{eq:criterion} ignores trading costs. Adding the linear costs
the investor actually pays would make it non-smooth. Following the
transaction-cost literature \citep{GarleanuPedersen2013, DeMiguel2020}, we instead
penalise a turnover proxy quadratically. Because the raw weight is linear in the predictors,
the change in factor $k$'s weight between two consecutive months of the window is
$\theta_k'(z_s-z_{s-1})$, holding $\theta$ fixed.\footnote{Realised turnover
also reflects monthly re-estimation. Costs in
evaluation are charged on changes in the final weights.} It depends on the data and the coefficients but not on
earlier decisions, so the problem remains a static choice of $\theta$ in the sense
of \citet{BrandtSantaClara2006}. The turnover-penalised criterion is
\begin{equation}\label{eq:costcriterion}
f_{\kappa}(\theta)=f(\theta)+\frac{\omega\kappa}{2}\sum_{\substack{s\in\mathcal T_t\\s>\min\mathcal T_t}}
\sum_k \phi_k\bigl\{\theta_k'(z_s-z_{s-1})\bigr\}^2,
\end{equation}
where $\phi_k$ is the proportional trading cost of factor $k$ and $\kappa\ge0$ sets
the strength of the penalty; the quadratic penalty equals the linear cost at a
monthly raw-weight change of $2/\kappa$ for $\kappa>0$. The sum uses $T-1$ within-window pairs. The penalty adds a positive semidefinite block to the
prior precision of each factor's coefficients, so $f_\kappa$ is smooth and strongly
convex and Lemma~\ref{lem:certified} applies. The constant does not change
from month to month, so the penalty falls only on the timing coefficients. Setting
$\kappa=0$ recovers \eqref{eq:criterion}.

\section{Data and Comparators}
\label{sec:data}
We use monthly returns on the market (Mkt), size (SMB), value (HML),
investment (CMA), profitability (RMW) and momentum (UMD) factors, and the
one-month Treasury bill rate, from Kenneth French's data library. The market
return is the total return and the
other five time-series are returns on zero-investment long--short portfolios. The
sample runs from \SampleStart{} to \OOSEnd.

The predictors are of two kinds. The first is \NAnomaly{} monthly returns on
anomaly portfolios from the Open Source Asset Pricing database of
\citet{ChenZimmermann2022}. They enter as aggregate returns, the same for every
factor at a given date, which is why the characteristic form of the policy
cannot use them (Section~\ref{subsec:annihilation}). The second is \NTiming{}
factor-specific predictors in the spirit of \citet{Haddad2020} and include time-series and
cross-sectional momentum, a valuation measure based on past prices, reversal and
volatility, each computed for each factor from its own returns.
Appendix~\ref{app:signals} gives the definitions. Every predictor is
standardised using only data available at the date it is used to avoid look-ahead bias in the out-of-sample exercises.\footnote{Changing all data from a
given date onwards leaves every earlier weight unchanged.}

The first \NBurnIn{} months are used only for estimation. From then on the
investor re-estimates the policy each month on the preceding \TimingWindow{}
months and holds the resulting weights for one month, which gives \NOOS{} monthly
returns from \OOSStart{} to \OOSEnd. The market weight is always one and the
benchmark factor weights are zero. Net returns deduct factor-specific trading and
holding costs, described in Appendix~\ref{sec:estimation}.

We compare BPPP and RPPP with seven conventional rules. Like BPPP, each holds the
market with a weight of one, chooses the five factor positions every month from
data available at the time, and faces the same position and gross-exposure
limits. Six rules use only past factor returns and ignore the predictors. The
\emph{constant allocation} holds the constant factor weights that maximise the
investor's CRRA utility, net of holding costs, over all past returns and subject
to both the position and gross-exposure limits; it is the
static benchmark against which timing must be judged. \emph{Minimum variance}
chooses the factor weights that minimise the variance of the whole portfolio,
market included, using the covariance matrix of the past 120 months, and
\emph{scaled minimum variance} rescales those weights to a bound on gross exposure
chosen to minimise variance, as in
\citet{DeMiguelGarlappiNogalesUppal2009}. \emph{Inverse variance} and \emph{reward
to risk} set weights proportional to each factor's inverse variance, or to its
positive mean divided by its variance, over the past 60 months
\citep{KirbyOstdiek2012}. The \emph{lasso} estimates sparse mean--variance weights
by regressing a vector of ones on the past 120 months of factor returns with a lasso
penalty, which by \citet{BrittenJones1999} is proportional to the tangency
portfolio, with the penalty chosen by hold-out Sharpe ratio and the weights scaled
to a volatility target, in the spirit of \citet{AoLiZheng2019}. It too ignores the
predictors, so any timing it shows comes from re-estimating its weights on a
rolling window. Finally, \emph{conditional mean--variance} is the two-stage
alternative to a parametric policy that uses the same predictors as BPPP. It
forecasts each factor's next-month return by a ridge regression on the predictors,
averaged over the past 12 or 48 months, and then chooses the factor weights that
maximise mean--variance utility with $\gamma=\RiskAversion$, with the ridge penalty
and averaging window chosen every 24 months by past out-of-sample certainty
equivalent.

We report Sharpe ratios of total returns, certainty equivalents for CRRA utility
with $\gamma=\RiskAversion$, and annual turnover of factor positions. Every rule
holds the same unit position in the market, so subtracting the bill rate lowers
all Sharpe ratios by similar amounts. It changes no comparison that involves
BPPP, and Table~\ref{tab:performance} reports both versions.
Although no future data enter any portfolio, the exercise is retrospective in the sense that the
set of predictors was assembled after the fact, and we do not reconstruct
historical data vintages.

\section{Results}
\label{sec:results}
We report results for BPPP, the posterior-averaged weights of
\eqref{eq:bppp_weights}, and for RPPP, the weights implied by the posterior mode.
We also report several conventional rules, to show what they do rather than to
rank them. Every rule faces the same limits, a position limit of $c=\TimingCap$ on
each factor and a limit of $g=\TimingLeverage$ on gross factor exposure. A rule's
weights are first capped at $\pm c$ and then, if their absolute sum exceeds $g$,
scaled down in proportion. The constant allocation imposes both limits during
optimisation. Returns are gross of trading costs unless stated
otherwise; Section~\ref{subsec:timingcosts} introduces costs. The section is
organised around the three questions set out in the introduction.

\subsection{Are factors timed?}
\label{subsec:timing}
A constant position in a factor earns its premium without timing it. Timing means
holding more of a factor before it performs well. For factor $k$,
\begin{equation}\label{eq:timingdecomp}
\mathbb E[w_{k,t}F_{k,t+1}]
=\mathbb E[w_{k,t}]\,\mathbb E[F_{k,t+1}]
+\operatorname{Cov}(w_{k,t},F_{k,t+1}),
\end{equation}
where the first term is the static component of the return and the second the
timing component, as in \citet{Lo2008}. The decomposition is an identity, and we
sum both parts over the five factors. BPPP earns \TimingComponent{} percentage points a year from timing, with a
Newey--West $t$-statistic of \TimingT. Its active return relative to the market
is \TimingActive{} points, of which \TimingStatic{} points are static and
\TimingComponent{} points, or \TimingShare\%, are timing.
Figure~\ref{fig:cumulative} shows how the timing return accumulates over the
sample.

\begin{figure}[h]\centering
\includegraphics[width=\textwidth]{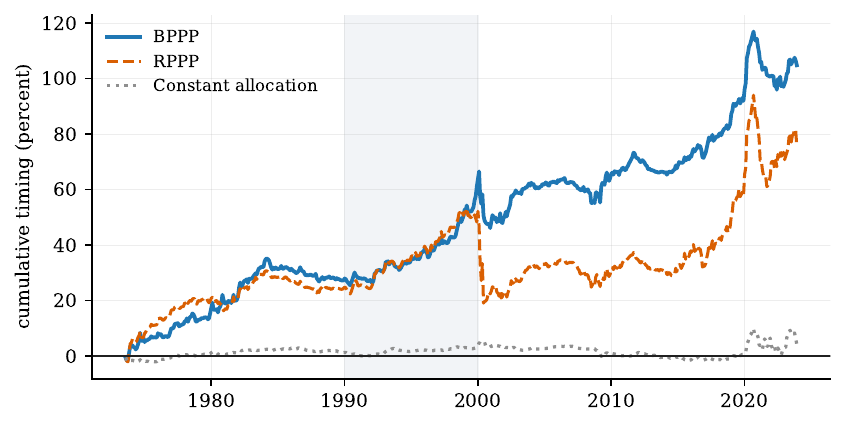}
\caption{Cumulative timing return}\label{fig:cumulative}
\begin{minipage}{\textwidth}\footnotesize
\textit{Notes:} Cumulative sum of the monthly timing return
$\sum_k(w_{k,t}-\bar w_k)(F_{k,t+1}-\bar F_k)$, in percent, for BPPP, RPPP and the
recursively estimated constant allocation. The constant allocation does not use
the predictors; its recursive re-estimation can nevertheless generate covariance
with subsequent returns. The shaded
area marks the 1990s.
\end{minipage}
\end{figure}

Table~\ref{tab:timing} applies the decomposition to the other rules. The constant
allocation, whose average positions are re-estimated each month but which ignores
the predictors, earns most of its active return from static exposure, and its timing
component is indistinguishable from zero. Minimum-variance and inverse-variance
weights change from month to month; their measured
timing components are \TimingMinvarTiming{} and \TimingInvvarTiming{} points. The
lasso, whose rolling re-estimation moves its weights over time, times
\TimingLassoTiming{} points ($t=\TimingLassoT$). Conditional mean--variance, which
uses the same predictors as BPPP but turns them into return forecasts first, times
only \TimingCmvTiming{} points ($t=\TimingCmvT$). RPPP times \TimingMapTiming{} points ($t=\TimingMapT$);
Section~\ref{subsec:mandate} explains why its timing is less precisely estimated
than that of BPPP.

\begin{table}[H]\centering
\caption{Static and timing components of active return}\label{tab:timing}
\small
\begin{tabular}{lrrrrrrr}
\toprule
& Active & Static & Timing & $t$ & Share & Gross & Exposure \\
& (pp) & (pp) & (pp) & & (\%) & SR & \\
\midrule
\textbf{BPPP} & 7.05 & 4.97 & 2.08 & 3.00 & 29.5 & 1.242 & 1.31 \\
\textbf{RPPP} & 6.77 & 5.27 & 1.50 & 1.53 & 22.1 & 1.211 & 1.49 \\
Constant allocation & 7.18 & 7.10 & 0.08 & 0.33 & 1.2 & 1.193 & 1.50 \\
Lasso & 5.59 & 4.77 & 0.82 & 1.00 & 14.7 & 1.094 & 1.40 \\
Conditional mean--variance & 6.17 & 5.19 & 0.98 & 1.15 & 15.9 & 1.160 & 1.50 \\
Minimum variance & 3.38 & 3.27 & 0.11 & 0.15 & 3.2 & 1.106 & 1.48 \\
Inverse variance & 3.55 & 3.68 & $-$0.12 & $-$0.72 & $-$3.5 & 1.002 & 1.00 \\
Reward to risk & 3.83 & 3.88 & $-$0.06 & $-$0.10 & $-$1.5 & 0.986 & 0.96 \\
Scaled minimum variance & 2.59 & 2.54 & 0.05 & 0.08 & 1.8 & 1.018 & 1.10 \\
Market & 0.00 & 0.00 & 0.00 & -- & -- & 0.739 & 0.00 \\
\bottomrule
\end{tabular}

\begin{minipage}{\textwidth}\footnotesize
\textit{Notes:} Annual percentage points, summed over the five factors,
\OOSStart{} to \OOSEnd. Active return is measured relative to the market. The
$t$-statistic uses a Newey--West estimate of the variance of the summed timing
return. Share is the timing component as a percentage of active return.
Exposure is the average sum of absolute factor weights. All rules face a position
limit of \TimingCap{} and a gross-exposure limit of \TimingLeverage. Returns are
gross of trading costs.
\end{minipage}
\end{table}

\subsection{Is the timing evidence robust?}
\label{subsec:timingrobust}
The reported prior scales were fixed. To understand how sensitive final results are to prior choice, we also set them recursively. Each month the investor considers nine combinations
of the static and timing scales, $\{0.4,0.8,1.6\}\times\{0.5,1.0,1.5\}$, and uses
the one whose out-of-sample returns up to the previous month have the highest
certainty equivalent.\footnote{The certainty equivalent is the constant return
that gives the investor the same expected utility as the risky portfolio,
$\mathrm{CE}=\{T^{-1}\sum_t(1+r_{p,t})^{1-\gamma}\}^{1/(1-\gamma)}-1$ per month
with $\gamma=\RiskAversion$, reported in annual terms. For the first 36 months,
before enough out-of-sample history exists, the centre of the grid is used. The
turnover of the selected policy includes changes of specification. The same rule
selects the scale of RPPP in Section~\ref{subsec:mandate} and, using net rather
than gross returns, the prior scales and turnover penalty in
Section~\ref{subsec:timingcosts}.} No future return enters any
choice. The resulting policy earns \RevRecTiming{} points a year from timing
($t=\RevRecT$), and it uses the reported scales in \RevRecModal\% of months.

We next check whether the Gaussian approximation behind the BPPP weights affects
the results. Estimating the exact posterior average by Monte Carlo at every date, as described in
Section~\ref{subsec:inference}, changes the weights by \RevExactMeanDiff{} on
average and by at most \RevExactMaxDiff, measured after the gross limit. Under the exact posterior the timing
component is \RevExactTiming{} points ($t=\RevExactT$), and the Sharpe ratio is
unchanged at the displayed precision for gross returns (Table~\ref{tab:exact}).

Timing is not spread evenly over time. The timing component is
\TimingSeventies{} points before 1980, \TimingEighties{} in the 1980s,
\TimingNineties{} in the 1990s, \TimingThousands{} in the 2000s, \TimingTens{} in
the 2010s and \TimingTwenties{} since 2020. It is positive in four of the six
periods, and the two exceptions are close to zero. Individual decades cannot be
expected to show significant timing on their own, with a full-sample
$t$-statistic of \TimingT{} spread over five decades. The more informative
question is whether one period drives the result. The 1990s contribute the most,
yet excluding them entirely still leaves \TimingExNineties{} points a year
($t=\TimingExNinetiesT$) over \TimingExNinetiesMonths{} months, and from 2000
onwards the estimate is \RevPost{} points, close to the full-sample value.

The anomaly portfolios were documented by academic research, and many of them
were unknown for much of the sample. When each anomaly return is used only from
the year its study was published, the timing component falls to \RevPub{} points
($t=\RevPubT$); with a two-year lag it is \RevPubTwo{} points ($t=\RevPubTwoT$).
The loss comes almost entirely from the period before 2000. After 2000, when almost
all predictors had been published, the restricted policy times \RevPubPost{} points,
nearly as much as the unrestricted policy's \RevPost. Timing on public information has therefore been
of similar size in the second half of the sample, while part of the full-sample
significance comes from anomalies used before they were documented. We regard the
restricted estimate as the better measure of what an investor could have achieved
in real time.

Across risk aversion $\gamma\in\{3,5,7,10\}$, estimation windows of 60, 120, 180
and 240 months, a fixed learning rate, and prior scales halved and doubled, the
timing component lies between \TimingRobustLow{} and \TimingRobustHigh{} points,
with $t$-statistics between \TimingRobustTLow{} and \TimingRobustTHigh.

\subsection{Which predictors time which factor?}
\label{subsec:whichsignals}
Removing any one of the five families of factor-specific predictors changes the
timing component by at most \RevFamilyMax{} points. Removing the \NAnomaly{}
anomaly returns instead reduces it to \RevDropAnomaly{} points
($t=\RevDropAnomalyT$). The aggregate predictors that the characteristic form of
the policy cannot use thus account for most of the timing.

The decomposition can also be split across predictors. Regressing factor $k$'s
demeaned weight on the demeaned predictors, $w_k-\bar w_k=\sum_\ell\beta_\ell
(z_\ell-\bar z_\ell)+e$, gives $\operatorname{Cov}(w_k,F_k)=\sum_\ell\beta_\ell
\operatorname{Cov}(z_\ell,F_k)+\operatorname{Cov}(e,F_k)$.
Table~\ref{tab:signalsbyfactor} in Appendix~\ref{sec:apptables} lists the ten
largest contributions for each factor. The value factor accounts for the largest
share of the timing component, followed by size and momentum, while the investment
and profitability factors contribute an order of magnitude less. Three of the
four largest contributions to the timing of value, in absolute value, come from
valuation-related anomalies, namely aggregate leverage, the asset-to-market ratio,
whose contribution is negative, and the sales-to-price ratio, which is consistent
with timing the value premium by the value spread, as in \citet{Haddad2020}. Individual contributions are small and
often offset one another, so timing is spread across many weak predictors rather
than a few strong ones. The attribution is descriptive, since it is computed in
sample and individual contributions are not separately identified when predictors
are correlated.

A natural alternative to shrinking \NSignals{} predictors is to compress them
first. \citet{Haddad2020} time the leading principal components of anomaly
returns, and dimension reduction is the usual response to a large predictor set.
We replace the predictors by their leading principal components, re-estimated at
each date from data available at that date, and estimate BPPP on them with the
same prior formula.\footnote{The sign of each component is fixed by making its largest
loading positive, so the basis is uniquely determined at each date. Because the
prior scale on the timing coefficients is divided by the number of predictors,
each component receives more prior variance than an individual predictor does.
PC scores are not variance-standardised, so the comparison also changes prior
predictive exposure and does not isolate dimension reduction.}
Table~\ref{tab:pca} reports the results. With \TimingPcaRank{} components, which
explain about \RevPcaTenExpl\% of the variance of the predictors, the timing
component rises to \RevPcaTen{} points ($t=\RevPcaTenT$), roughly twice that of the
full predictor set, and it exceeds the full-set estimate at every rank we
consider. With fewer predictors, each retained direction receives more prior variance,
and the estimated policies take more aggressive positions. The
cost is trading since annual turnover is between \RevPcaTurnLow{} and \RevPcaTurnHigh,
against \RevTurnover{} with all predictors. Net of costs, the principal-component policies earn Sharpe ratios of only
\RevPcaNetLow{} to \RevPcaNetHigh. These principal-component specifications therefore find more timing
in the data, while shrinkage across all predictors finds timing that is cheaper to
exploit. This distinction matters for any factor-timing strategy that is to be
implemented rather than measured.

\begin{table}[t]\centering
\caption{Timing with principal components of the predictors}\label{tab:pca}
\small
\begin{tabular}{lrrrrrr}
\toprule
& Variance & Timing & $t$ & Gross & Net & Turnover \\
& explained & (pp) & & SR & SR & \\
\midrule
3 components & 35\% & 2.57 & 2.46 & 1.031 & 0.564 & 15.43 \\
5 components & 43\% & 3.98 & 3.67 & 1.086 & 0.633 & 15.01 \\
10 components & 53\% & 4.42 & 4.51 & 1.169 & 0.734 & 14.13 \\
20 components & 65\% & 4.17 & 4.34 & 1.183 & 0.774 & 12.77 \\
40 components & 77\% & 3.57 & 4.13 & 1.204 & 0.849 & 10.61 \\
\addlinespace
All 242 predictors & -- & 2.08 & 3.00 & 1.242 & 1.043 & 5.01 \\
\bottomrule
\end{tabular}

\begin{minipage}{\textwidth}\footnotesize
\textit{Notes:} BPPP estimated on the leading principal components of the
\NSignals{} standardised predictors, re-estimated at each date. Variance explained
is the median share of predictor variance captured by the components. Timing in
annual percentage points with its Newey--West $t$-statistic. Net Sharpe ratios
deduct the trading and holding costs of Appendix~\ref{sec:estimation}. Turnover is
the annual sum of absolute changes in factor weights.
\end{minipage}
\end{table}

\subsection{What survives a leverage limit?}
\label{subsec:mandate}
Without a common limit on gross factor exposure, differences between rules would
partly reflect differences in the size of their positions. The limit affects BPPP
and RPPP differently. RPPP reaches it in \RevAtLimitMap\% of months and BPPP in
\RevAtLimitBPPP\%. Proportional scaling keeps the direction of a tilt, but it
scales positions down most in the months when they are largest, which compresses
the variation over time that the timing component measures. RPPP times
\TimingMapTiming{} points ($t=\TimingMapT$), against \TimingComponent{}
($t=\TimingT$) for BPPP. Table~\ref{tab:mechanism} and Figure~\ref{fig:frontier}
ask where the difference comes from, by comparing BPPP with versions of RPPP whose
positions have been made smaller in different ways.

\begin{table}[p]\centering
\caption{BPPP and smaller versions of RPPP}\label{tab:mechanism}
\footnotesize
\begin{tabular}{lrrrrrrr}
\toprule
& Exposure & At limit & Timing & $t$ & Gross & $\Delta$ timing & $t$ \\
& & (\%) & (pp) & & SR & vs BPPP & \\
\midrule
\textbf{BPPP} & 1.31 & 52.6 & 2.08 & 3.00 & 1.242 & -- & -- \\
\textbf{RPPP} & 1.49 & 96.4 & 1.50 & 1.53 & 1.211 & $-$0.58 & $-$0.98 \\
\addlinespace
RPPP, prior scales $\times0.8$ & 1.48 & 88.9 & 1.52 & 1.65 & 1.220 & $-$0.56 & $-$1.14 \\
RPPP, prior scales $\times0.6$ & 1.41 & 75.9 & 1.60 & 1.86 & 1.226 & $-$0.48 & $-$1.41 \\
RPPP, prior scales $\times0.5$ & 1.34 & 62.5 & 1.69 & 2.08 & 1.222 & $-$0.39 & $-$1.58 \\
RPPP, prior scales $\times0.4$ & 1.12 & 17.0 & 1.47 & 2.08 & 1.162 & $-$0.61 & $-$2.34 \\
RPPP, timing scale $0.5$ & 1.49 & 94.4 & 1.10 & 1.30 & 1.206 & $-$0.98 & $-$1.72 \\
RPPP, timing scale $0.2$ & 1.49 & 91.2 & 0.93 & 1.17 & 1.190 & $-$1.15 & $-$1.92 \\
\addlinespace
Bagged ridge & 1.47 & 86.6 & 1.69 & 1.95 & 1.241 & $-$0.39 & $-$0.91 \\
\addlinespace
RPPP $\times0.375$, matched exposure$^\dagger$ & 1.31 & 53.9 & 2.04 & 2.65 & 1.237 & $-$0.04 & $-$0.25 \\
RPPP at BPPP exposure each month$^\ddagger$ & 1.31 & 52.6 & 1.55 & 2.19 & 1.201 & $-$0.53 & $-$1.34 \\
\bottomrule
\end{tabular}

\begin{minipage}{\textwidth}\footnotesize
\textit{Notes:} All rules face the limits of Section~\ref{sec:results}. At limit
is the percentage of months in which gross factor exposure equals the limit.
$\Delta$ timing is the rule's timing component minus that of BPPP, with the
Newey--West $t$-statistic of the monthly difference. ``Prior scales $\times a$''
multiplies both $\delta_{\mathrm s}$ and $\delta_\tau$ by $a$; ``timing scale''
changes $\delta_\tau$ only. Bagged ridge is described in
Section~\ref{subsec:whybayes}. $^\dagger$The scale is chosen on the full sample to
match the average gross exposure of BPPP, which requires hindsight.
$^\ddagger$RPPP is scaled each month to the gross exposure of BPPP in that month;
this uses BPPP's weights and is not a feasible rule.
\end{minipage}
\end{table}

\begin{figure}[t]\centering
\includegraphics[width=\textwidth]{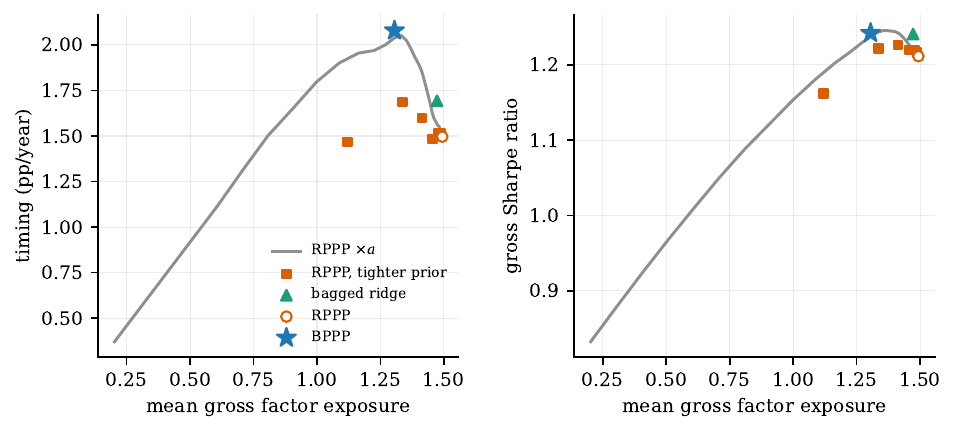}
\caption{Timing and gross exposure}\label{fig:frontier}
\begin{minipage}{\textwidth}\footnotesize
\textit{Notes:} The line multiplies RPPP by a constant $a\in[0.05,1]$ before the
limits are applied, and plots the timing component (left) and the gross Sharpe
ratio (right) against average gross factor exposure. Squares tighten the prior of
RPPP (both prior scales multiplied by 0.4 to 0.9). The triangle is bagged ridge,
the circle RPPP and the star BPPP.
\end{minipage}
\end{figure}

Three findings emerge. First, the size of the positions determines how much timing
survives the limit, and BPPP chooses that size well. The line in
Figure~\ref{fig:frontier} scales RPPP by every constant between 0.05 and 1. Timing
is highest, at \RevPeakTiming{} points, when RPPP is multiplied by
\RevPeakScaleTiming, and the Sharpe ratio is highest, at \RevPeakSR, at
\RevPeakScaleSR. BPPP lies near the top of this line, with \TimingComponent{} points
and a Sharpe ratio of \TimingGrossSR, although it never observes the line.

Second, BPPP chooses the size without a separate exposure-scale parameter. By
Proposition~\ref{prop:shrink}(iii), averaging keeps a fraction of about
$2\Phi(c/\sigma)-1$ of each raw mode position locally, before gross scaling, and in the data the BPPP weights have a
correlation of \RevContractionCorr{} with this approximation. At the median factor
and date the local slope is \RevContraction, and 80\% of factor-months lie between
\RevContractionLow{} and \RevContractionHigh. Reaching the same point with RPPP
requires choosing a scale from the data. With hindsight, the scale that matches
BPPP's average exposure, \RevScaleMatched, gives \RevScaledTiming{} points
($t=\RevScaledT$); without hindsight, a scale chosen each month from the past
out-of-sample performance of eight candidate scales gives \RevRecScaledTiming{}
points ($t=\RevRecScaledT$). BPPP reaches similar timing without a separate scale search, and it also varies
the shrinkage across factors and dates. Scaling RPPP each month to BPPP's own gross
exposure, which is not a feasible rule, gives \RevMonthlyTiming{} points
($t=\RevMonthlyT$), which suggests that allocating the shrinkage across factors by
their uncertainty adds to the timing of BPPP, although the difference is not
statistically significant.

Third, the familiar ways to shrink a point estimate without choosing a scale by
hand do not reach the same point. A tighter prior, which is a larger ridge penalty,
reduces timing to between \RevRidgeTimingLow{} and \RevRidgeTimingHigh{} points. It
shrinks the timing coefficients but leaves the average positions pressed against
the limit. With both prior scales halved, the policy is still at the limit in
\RevRidgeHalfAtLimit\% of months, and tightening further moves it inside the limit
only once timing is significantly lower than that of BPPP. Bagged ridge averages
within the limit, but its bootstrap standard deviation of a factor weight is
\RevBootSD{}, against a posterior standard deviation of \RevPostSD, because
resampling produces no variation in the directions the window cannot identify. It
therefore shrinks too little, stays at the limit in \RevBagAtLimit\% of months, and
times \RevBagTiming{} points ($t=\RevBagT$).

\subsection{Timing and trading costs}
\label{subsec:timingcosts}
Timing requires trading, and trading is costly. We compare two ways to trade
less. The first keeps the original criterion and tightens the prior on the timing
coefficients. The second adds the turnover penalty of \eqref{eq:costcriterion}.

\begin{figure}[t]\centering
\includegraphics[width=\textwidth]{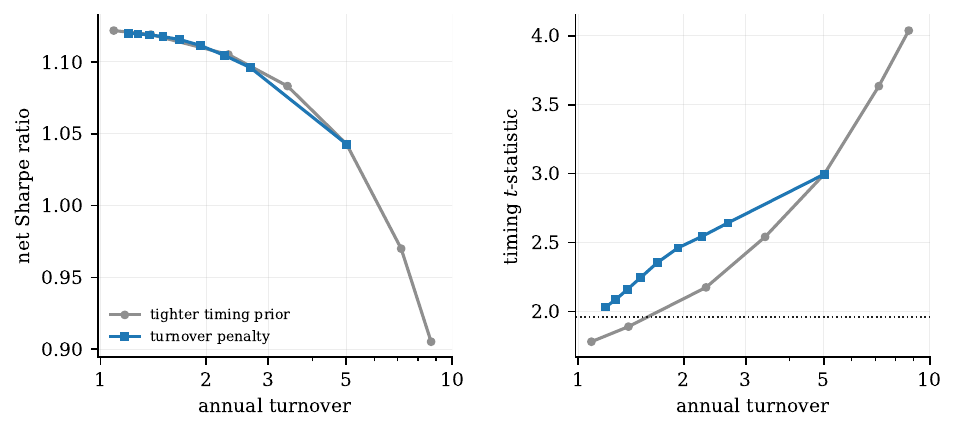}
\caption{Two ways to trade less}
\label{fig:costs}
\begin{minipage}{\textwidth}\footnotesize
\textit{Notes:} Grey: BPPP with timing prior scale
$\delta_\tau\in\{0.2,0.3,0.5,0.7,1,1.5,2\}$. Blue: BPPP estimated with the
turnover penalty \eqref{eq:costcriterion},
$\kappa\in\{0,5,10,20,40,80,160,320,640\}$. Left: Sharpe ratio net of the costs in
Appendix~\ref{sec:estimation}. Right: $t$-statistic of the timing component; the dotted line is 1.96.
\end{minipage}
\end{figure}

\begin{table}[p]\centering
\caption{Turnover-penalised estimation}\label{tab:costaware}
\small
\begin{tabular}{lrrrrrrr}
\toprule
& Turnover & Timing & $t$ & Gross & Net & Net SR & Net CE \\
& & (pp) & & SR & SR & $p$ vs const. & (\%) \\
\midrule
BPPP, no penalty & 5.01 & 2.08 & 3.00 & 1.242 & 1.043 & 0.21 & 9.64 \\
BPPP, $\kappa=5$ & 2.67 & 1.66 & 2.64 & 1.221 & 1.096 & 0.56 & 10.49 \\
BPPP, $\kappa=10$ & 2.25 & 1.59 & 2.54 & 1.218 & 1.105 & 0.65 & 10.62 \\
BPPP, $\kappa=20$ & 1.92 & 1.53 & 2.46 & 1.216 & 1.111 & 0.72 & 10.73 \\
BPPP, $\kappa=40$ & 1.68 & 1.47 & 2.36 & 1.214 & 1.116 & 0.76 & 10.79 \\
BPPP, $\kappa=80$ & 1.51 & 1.42 & 2.25 & 1.212 & 1.118 & 0.79 & 10.83 \\
BPPP, $\kappa=160$ & 1.38 & 1.38 & 2.16 & 1.210 & 1.119 & 0.80 & 10.84 \\
BPPP, $\kappa=320$ & 1.28 & 1.35 & 2.09 & 1.208 & 1.119 & 0.81 & 10.85 \\
BPPP, $\kappa=640$ & 1.20 & 1.33 & 2.03 & 1.206 & 1.120 & 0.82 & 10.85 \\
\addlinespace
Recursively selected$^\dagger$ & 1.21 & 1.26 & 1.93 & 1.208 & 1.121 & 0.83 & 10.84 \\
Constant allocation & 0.15 & 0.08 & 0.33 & 1.193 & 1.136 & -- & 11.18 \\
\bottomrule
\end{tabular}

\begin{minipage}{\textwidth}\footnotesize
\textit{Notes:} BPPP estimated with the turnover penalty \eqref{eq:costcriterion}
at the reported prior scales. Timing and its $t$-statistic are gross of costs. Net returns deduct the trading and holding costs of
Appendix~\ref{sec:estimation}. The $p$-value is for the difference in net Sharpe
ratios from the constant allocation, using the HAC test of \citet[Section~3.1]{LedoitWolf2008}.
$^\dagger$Prior scales and $\kappa$ are chosen each month from a $3\times3\times3$
grid by past out-of-sample net certainty equivalent.
\end{minipage}
\end{table}

Figure~\ref{fig:costs} shows the result. Both methods raise the net Sharpe ratio
as they lower turnover, along the same frontier. The two methods differ in the
timing they keep. With the turnover penalty the timing component remains
significant at every value of $\kappa$, with $t$-statistics of at least
\RevCostTLow; at $\kappa=80$ the policy turns over \RevCostEightyTurn{} times a year,
against \RevTurnover{} without the penalty, and times \RevCostEightyTiming{} points
($t=\RevCostEightyT$). With the tightest prior the policy turns over
\RevScaleTwoTurn{} times a year and reaches a similar net Sharpe ratio of
\RevScaleTwoNet, but its timing component of \RevScaleTwoTiming{} points
($t=\RevScaleTwoT$) is no longer significant. The penalty acts on each coefficient
in proportion to how much its predictor moves from month to month, whereas a
tighter prior shrinks all timing coefficients alike. The comparison suggests that
coefficients on fast-moving predictors generate turnover without much timing, while
the slow-moving predictors carry it.

The benchmark for judging timing net of costs is demanding. The constant
allocation is not a naive $1/N$ portfolio but the static factor tilt that
maximises the investor's utility, re-estimated each month, and it almost never
trades. Against this benchmark the turnover-penalised policy holds its own
(Table~\ref{tab:costaware}). At $\kappa=80$ its net Sharpe ratio is
\RevCostEightyNet{} against \RevConstNetSR{} ($p=\RevCostEightyP$), and when the
prior scales and $\kappa$ are all chosen recursively it is \RevRecCostNetSR{}
($p=\RevRecCostP$), with a difference in net certainty equivalents of
\RevRecCostCEDiff{} percentage points a year. Net of costs, the recursively chosen
policy earns an insignificant alpha of \RevRecAlpha{} percentage points
a year relative to the constant allocation ($t=\RevRecAlphaT$). Its gross timing
has $t=\RevRecCostT$, below the conventional cutoff. Combining the
two with hindsight, with \RevBlendW\% in the timing policy, gives a slightly higher net Sharpe
ratio, \RevBlendSR, than either alone.

Whether timing pays is therefore a question of implementation costs. We apply the
same cost schedule to every rule, and it is deliberately conservative, charging
60 to 70 basis points to trade the size and momentum factors and 25 to 30 basis
points for the others, plus holding costs for the short leg. The turnover-penalised
policy has a higher estimated net Sharpe ratio than the constant allocation when these costs are scaled
below about \RevBreakEighty\% of our schedule (\RevBreakRec\% for the recursively
chosen policy). These are point-estimate crossovers, not thresholds for
statistically significant superiority. \citet{FrazziniIsraelMoskowitz2012} find that the costs a large
institutional trader actually pays are well below standard academic estimates. For
investors who trade at such costs, the point estimates suggest value from timing beyond the static
allocation. For investors who pay our conservative costs, the net Sharpe difference is not statistically significant, which is consistent with the scepticism of \citet{Arnott2019} and
\citet{Ilmanen2021}.

\subsection{Risk-adjusted performance}
\label{subsec:performance}
Table~\ref{tab:performance} reports conventional performance statistics, complementing the timing decomposition. Before costs, BPPP has the highest Sharpe ratio in the table, \TimingGrossSR,
against \TimingMarketSR{} for the market. BPPP, RPPP, the turnover-penalised BPPP
and the constant allocation form a group with gross Sharpe ratios between
\TimingConstSR{} and \TimingGrossSR, well above the variance-based rules, and
paired tests with a Holm correction do not distinguish BPPP from the constant
allocation ($p=\TimingPConst$) or from RPPP ($p=\TimingPMap$). The decomposition
in Table~\ref{tab:timing} is what separates these rules. They earn similar Sharpe
ratios, one predominantly from static positions and the others partly from timing.
After costs the ordering within the group follows turnover, as
Section~\ref{subsec:timingcosts} explains. Subtracting the bill rate lowers the
gross Sharpe ratio of BPPP to \RevExcessGrossSR{} and that of the constant
allocation to \RevExcessConstGross{}, and leaves these comparisons unchanged.

\begin{table}[p]\centering
\caption{Portfolio performance, \OOSStart{} to \OOSEnd}
\label{tab:performance}
\footnotesize
\begin{tabular}{lrrrrrrr}
\toprule
& \multicolumn{2}{c}{Gross SR} & \multicolumn{2}{c}{Net SR} & Gross CE & Turnover & Exposure \\
\cmidrule(lr){2-3}\cmidrule(lr){4-5}
& Total & Excess & Total & Excess & (\%) & & \\
\midrule
\textbf{RPPP} & 1.211 & 0.931 & 0.992 & 0.714 & 12.26 & 5.49 & 1.49 \\
\textbf{BPPP} & 1.242 & 0.960 & 1.043 & 0.761 & 12.68 & 5.01 & 1.31 \\
BPPP, turnover-penalised$^\dagger$ & 1.208 & 0.926 & 1.121 & 0.839 & 12.18 & 1.21 & 1.39 \\
Constant allocation & 1.193 & 0.925 & 1.136 & 0.869 & 12.09 & 0.15 & 1.50 \\
Lasso & 1.094 & 0.824 & 0.993 & 0.724 & 10.63 & 1.88 & 1.40 \\
Conditional mean--variance & 1.160 & 0.882 & 1.044 & 0.768 & 11.72 & 2.00 & 1.50 \\
Minimum variance & 1.106 & 0.792 & 1.027 & 0.713 & 10.37 & 0.68 & 1.48 \\
Inverse variance & 1.002 & 0.721 & 0.961 & 0.680 & 9.09 & 0.27 & 1.00 \\
Reward to risk & 0.986 & 0.715 & 0.920 & 0.649 & 9.03 & 1.20 & 0.96 \\
Scaled minimum variance & 1.018 & 0.714 & 0.954 & 0.650 & 9.25 & 0.73 & 1.10 \\
Market & 0.739 & 0.470 & 0.739 & 0.470 & 5.04 & 0.00 & 0.00 \\
\bottomrule
\end{tabular}

\begin{minipage}{\textwidth}\footnotesize
\textit{Notes:} \NOOS{} monthly returns; all rules face the limits of
Section~\ref{sec:results}. Gross returns deduct no costs; net returns deduct the
trading and holding costs of Appendix~\ref{sec:estimation}. Total Sharpe ratios
use total returns; excess Sharpe ratios subtract the one-month bill rate. CE is
the annualised CRRA certainty equivalent with $\gamma=\RiskAversion$. Turnover is
the annual sum of absolute changes in factor weights. $^\dagger$BPPP with the
turnover penalty, with prior scales and $\kappa$ chosen recursively
(Table~\ref{tab:costaware}).
\end{minipage}
\end{table}

\section{Robustness and Limits}
\label{sec:robustness}
Section~\ref{subsec:mandate} rests on limits we chose, so
Table~\ref{tab:mandatesensitivity} varies the position limit over 0.4, 0.6, 0.8
and 1.0 and the gross-exposure limit over 1.0, 1.5, 2.0 and 3.0. The timing
$t$-statistic of BPPP lies between \RevMandateBTLow{} and \RevMandateBTHigh{} for
all sixteen combinations, while that of RPPP lies between \RevMandateMTLow{} and
\RevMandateMTHigh. The size of the timing component grows with the permitted
exposure, from \RevMandateTimLow{} points under the tightest limits to
\RevMandateTimHigh{} under the loosest, and net Sharpe ratios vary much less, from
\RevMandateNetLow{} to \RevMandateNetHigh, because larger positions earn more from
timing and pay more in costs. The reported combination does not maximise any of
these statistics; we chose it because gross factor exposure of \TimingLeverage{}
on top of the market is a plausible limit for an institutional setting.

\begin{table}[p]\centering
\caption{Sensitivity to the limits}\label{tab:mandatesensitivity}
\footnotesize
\begin{tabular}{llrrrr}
\toprule
& & \multicolumn{4}{c}{Gross-exposure limit} \\
\cmidrule(lr){3-6}
Position limit & & 1.0 & 1.5 & 2.0 & 3.0 \\
\midrule
0.4 & BPPP gross SR & 1.102 & 1.146 & 1.146 & 1.146 \\
& BPPP net SR & 0.966 & 0.992 & 0.992 & 0.992 \\
& BPPP timing ($t$) & 1.39 (2.99) & 1.61 (3.05) & 1.61 (3.05) & 1.61 (3.05) \\
& RPPP timing ($t$) & 0.93 (1.39) & 1.43 (1.52) & 1.80 (1.69) & 1.80 (1.69) \\
& BPPP $-$ RPPP ($t$) & 0.46 (1.06) & 0.18 (0.28) & $-$0.19 ($-$0.25) & $-$0.19 ($-$0.25) \\
\addlinespace
0.6 & BPPP gross SR & 1.120 & 1.242 & 1.284 & 1.287 \\
& BPPP net SR & 0.965 & 1.043 & 1.065 & 1.066 \\
& BPPP timing ($t$) & 1.39 (2.49) & 2.08 (3.00) & 2.37 (3.06) & 2.38 (3.05) \\
& RPPP timing ($t$) & 0.99 (1.48) & 1.50 (1.53) & 2.03 (1.67) & 2.77 (1.90) \\
& BPPP $-$ RPPP ($t$) & 0.40 (1.15) & 0.58 (0.98) & 0.34 (0.43) & $-$0.39 ($-$0.40) \\
\addlinespace
0.8 & BPPP gross SR & 1.122 & 1.253 & 1.341 & 1.375 \\
& BPPP net SR & 0.963 & 1.034 & 1.088 & 1.101 \\
& BPPP timing ($t$) & 1.38 (2.23) & 2.09 (2.65) & 2.76 (3.01) & 3.09 (3.05) \\
& RPPP timing ($t$) & 1.01 (1.52) & 1.55 (1.59) & 2.08 (1.69) & 3.25 (2.03) \\
& BPPP $-$ RPPP ($t$) & 0.37 (1.42) & 0.54 (1.16) & 0.68 (1.01) & $-$0.16 ($-$0.17) \\
\addlinespace
1.0 & BPPP gross SR & 1.125 & 1.255 & 1.343 & 1.414 \\
& BPPP net SR & 0.963 & 1.027 & 1.073 & 1.107 \\
& BPPP timing ($t$) & 1.41 (2.19) & 2.08 (2.42) & 2.82 (2.78) & 3.70 (3.05) \\
& RPPP timing ($t$) & 1.07 (1.63) & 1.64 (1.70) & 2.23 (1.82) & 3.55 (2.18) \\
& BPPP $-$ RPPP ($t$) & 0.33 (1.68) & 0.44 (1.29) & 0.59 (1.15) & 0.15 (0.20) \\
\bottomrule
\end{tabular}

\begin{minipage}{\textwidth}\footnotesize
\textit{Notes:} Each block is a position limit and each column a gross-exposure
limit. Timing is in annual percentage points with its Newey--West $t$-statistic
in parentheses. BPPP $-$ RPPP is the difference in timing between the two, with the
$t$-statistic of the monthly difference. The reported specification has a
position limit of \TimingCap{} and a gross-exposure limit of \TimingLeverage. When
a gross-exposure limit exceeds the exposure a policy takes, its column repeats the
previous one.
\end{minipage}
\end{table}

The comparison between BPPP and RPPP deserves a precise statement. BPPP times more
than RPPP in twelve of the sixteen combinations, including every combination with a
gross-exposure limit of 1.5 or tighter, and its timing is estimated more precisely
in all of them. The monthly difference between the two is not statistically
significant under any combination, however, and when the gross-exposure limit is
loose the two are similar. The same holds against RPPP scaled
down to BPPP's exposure with hindsight (difference $t=\RevScaledDiffT$). The
evidence therefore supports the claim of Section~\ref{subsec:mandate}, that BPPP
chooses a good position size without a separate exposure-scale search, rather than a claim that averaging
over the posterior produces timing that no point estimate could.

A few further words of caution are in order. The real-time evidence is weaker than
the full-sample evidence. With the anomaly predictors used only after publication,
the timing component is \RevPub{} points, although the estimate
after 2000 is about \RevPost{} points in every version. The
benchmark matters, since a benchmark that already held the factors would move part
of the static component into the benchmark and change what the decomposition
counts as active. The numerical checks bound the optimisation error of every
estimate and compare the Gaussian approximation with the
Monte Carlo estimates under the exact posterior at every date, but they do not validate the calibration of the
learning rate or the choice of prior. Finally, the set of predictors was assembled
aposteriori. This is, by choosing the prior and the turnover penalty recursively we remove the
most direct way in which knowledge of the full sample could affect the results, but
it does not reconstruct historical data vintages or provide an untouched test. We do not think this compromises the study but it is a limitation worth reporting.

\section{Conclusion}
\label{sec:conclusion}
The conditional portfolio policies of \citet{BrandtSantaClara2006} let each
factor's weight respond to aggregate predictors. With the predictors proposed in
the factor-timing literature, the policy has many more coefficients than any
estimation window has observations. We develop Bayesian Parametric Portfolio
Policies, which estimate such a policy with a generalised Bayesian posterior and
average its weights over the posterior, and use them to study three questions.

First, factors are timed. The policy earns \TimingComponent{} percentage points a
year from the covariance between its weights and subsequent factor returns
($t=\TimingT$), and \RevRecTiming{} points ($t=\RevRecT$) when the prior is chosen
recursively. Most of this comes from aggregate anomaly returns, and the value
factor accounts for the largest share. The evidence is weaker in real time, since
using each predictor only after its publication reduces the estimate to \RevPub{}
points, but timing on public information after 2000 is of similar
size to the full-sample estimate.

Second, a leverage limit determines how much timing survives. Proportional
deleveraging scales positions down most when they are largest and so compresses
their variation over time. Averaging over the posterior scales each weight back by
an amount that depends on its uncertainty relative to the position limit, and it
arrives, without a separate exposure-scale choice, near the position size that
maximises timing with hindsight. A heavier ridge penalty and bootstrap aggregation do not.

Third, a turnover penalty preserves significant gross timing at a fraction of
the turnover. The penalised policy's net Sharpe ratio is statistically
indistinguishable from that of the static factor allocation at the full cost
schedule, and its point estimate is higher when costs are below about half of it. Whether factor timing pays is therefore a question
of who implements it. More broadly, the interaction between estimation uncertainty
and leverage limits that we document can also arise in other timing strategies
that operate under position and gross-exposure limits.

\clearpage
\setlength\bibsep{0pt}
\bibliographystyle{apalike}
\bibliography{library}

\newpage
\appendix
\section{Proofs}
\label{sec:proofs}
\subsection*{Proof of Proposition~\ref{prop:shrink}}
Write $\varphi$ and $\Phi$ for the standard normal density and distribution
function, and
\[
a=\frac{-c-\mu}{\sigma},\qquad b=\frac{c-\mu}{\sigma}
\]
for the two limits expressed in standard deviations from the mean, so that
$\Pr(X<-c)=\Phi(a)$, $\Pr(X>c)=1-\Phi(b)$ and $\Pr(-c\le X\le c)=\Phi(b)-\Phi(a)$.
We take $\mu\ge0$. The case $\mu<0$ follows by symmetry, because the limit is odd,
$\Pi_c(-x)=-\Pi_c(x)$, so replacing $\mu$ by $-\mu$ changes the sign of $\bar w$
and nothing else.

\emph{Step 1:} The constrained weight $\Pi_c(X)$ equals $-c$ when
$X<-c$, equals $X$ when $-c\le X\le c$, and equals $c$ when $X>c$. Averaging over
these three regions, and using the standard result
$\mathbb E[X\,\mathbf 1\{-c\le X\le c\}]=\mu\{\Phi(b)-\Phi(a)\}+\sigma\{\varphi(a)-\varphi(b)\}$
for the middle one, gives
\begin{equation}\label{eq:wbar}
\bar w=f(\mu,\sigma)=-c\,\Phi(a)+\mu\{\Phi(b)-\Phi(a)\}
+\sigma\{\varphi(a)-\varphi(b)\}+c\{1-\Phi(b)\}.
\end{equation}
which is a closed form.

\emph{Step 2: how $\bar w$ responds to $\mu$ and to $\sigma$.} Differentiating
\eqref{eq:wbar} and simplifying with $\varphi'(x)=-x\varphi(x)$, $da/d\mu=db/d\mu=-1/\sigma$,
$da/d\sigma=-a/\sigma$ and $db/d\sigma=-b/\sigma$ gives
\begin{equation}\label{eq:wbarmu}
\frac{\partial f}{\partial\mu}=\Phi(b)-\Phi(a)=\Pr(-c\le X\le c),
\qquad
\frac{\partial f}{\partial\sigma}=\varphi(a)-\varphi(b).
\end{equation}
The first derivative has a simple interpretation. Shifting the whole distribution
of $X$ up by a small amount moves the draws that lie inside the limit by the same
amount, while the draws beyond the limit stay at $\pm c$. The average constrained
weight therefore moves by the share of draws inside the limit, a number strictly
between zero and one.

\emph{Part (i).} If $\mu=0$ the distribution of $X$ is symmetric about zero, so
$\bar w=0$. For $\mu>0$, \eqref{eq:wbarmu} says that $\bar w$ rises with $\mu$ but
more slowly than $\mu$, so $0<\bar w<\mu$. In addition $\bar w<c$, because
$\Pi_c(X)\le c$ always and $\Pi_c(X)<c$ whenever $X<c$, which happens with positive
probability. Together, $0<\bar w<\min(\mu,c)=\Pi_c(\mu)$.

\emph{Part (ii).} For $\mu>0$ the lower limit is further from the mean than the
upper one, $|a|>|b|$, so $\varphi(a)<\varphi(b)$ and, by \eqref{eq:wbarmu},
$\bar w$ falls as $\sigma$ rises. Intuitively, spreading the distribution pushes
more probability beyond the upper limit, where it is cut back to $c$, than beyond
the lower limit, so the average constrained weight falls. As $\sigma\to0$ the
distribution collapses on $\mu$ and $\bar w\to\Pi_c(\mu)$.

\emph{Part (iii).} We expand $f$ in $\mu$ around zero. At $\mu=0$ we have $a=-b$,
so $f=0$, $\partial f/\partial\mu=\Phi(c/\sigma)-\Phi(-c/\sigma)=2\Phi(c/\sigma)-1=\psi(\sigma)$,
and $\partial^2 f/\partial\mu^2=\{\varphi(a)-\varphi(b)\}/\sigma=0$. The third
derivative is $\{\varphi'(b)-\varphi'(a)\}/\sigma^2$, and since
$|\varphi'(x)|=|x|\varphi(x)$ is largest at $|x|=1$, it is at most
$2\varphi(1)/\sigma^2$ in absolute value. Taylor's theorem then gives
$\bar w=\psi(\sigma)\mu+\varepsilon$ with
\[
|\varepsilon|\le\frac{2\varphi(1)}{\sigma^2}\,\frac{|\mu|^3}{6}
=\frac{\varphi(1)}{3}\,\frac{|\mu|^3}{\sigma^2}\le0.081\,\frac{|\mu|^3}{\sigma^2},
\]
because $\varphi(1)=0.242$. The approximation is accurate when $|\mu|$ is not large
relative to $\sigma$, which is the case in our application.

\emph{Part (iv).} Write
$\Pi_c(X)=c\,\mathbf 1\{X>c\}-c\,\mathbf 1\{X<-c\}+X\,\mathbf 1\{|X|\le c\}$. The
last term is at most $c$ in absolute value and is non-zero only with probability
$\Pr(|X|\le c)$. Hold $s=\mu/\sigma$ fixed and let $\sigma$ grow, so that
$c/\sigma\to0$. Then $\Pr(X>c)=\Phi(s-c/\sigma)\to\Phi(s)=\Pr(X>0)$,
$\Pr(X<-c)=\Phi(-s-c/\sigma)\to\Phi(-s)=\Pr(X<0)$ and $\Pr(|X|\le c)\to0$.
Dividing $\bar w$ by $c$ and taking limits gives
$\bar w/c\to\Pr(X>0)-\Pr(X<0)=\Phi(s)-\Phi(-s)=2\Phi(s)-1$. \hfill$\square$

\subsection*{A unique mode and a bound on its error}
The affine criterion has a unique global minimiser, and the numerical objective
error can be bounded. We use the following standard property of strongly convex functions
\citep[Section~9.1.2]{BoydVandenberghe2004}.

\begin{lemma}[Mode Uniqueness]\label{lem:certified}
Assume $\omega>0$, $\gamma>1$, $\Lambda\succ0$, and that the positive-wealth
domain $\Theta_+=\{\theta:q_s(\theta)>0\ \forall s\}$ is nonempty. Then the
criterion $f$ in \eqref{eq:criterion} has exactly one minimiser on $\Theta_+$,
the posterior mode, and for every $\theta\in\Theta_+$,
\[
f(\theta)-\min_{\theta'}f(\theta')\le
\tfrac12\,\nabla f(\theta)'\Lambda^{-1}\nabla f(\theta).
\]
The same holds for the turnover-penalised criterion \eqref{eq:costcriterion},
with $\Lambda$ replaced by $\Lambda+M$, where
$M=\omega\kappa\,\operatorname{diag}_k(\phi_k\,\Delta Z'\Delta Z)$, a block-diagonal
matrix whose $k$th block is $\phi_k\,\Delta Z'\Delta Z$ scaled by $\omega\kappa$.
\end{lemma}

\noindent\emph{Proof.} The criterion has two parts. The utility part,
$-\omega\sum_sU\{q_s(\theta)\}$, is convex since wealth $q_s$ is affine in $\theta$ by
\eqref{eq:affinewealth}, and utility is increasing and concave, so minus utility
is convex in $\theta$. The prior part, $\tfrac12\theta'\Lambda\theta$, is a bowl
whose curvature is $\Lambda$ in every direction. Adding the two gives a function
that curves upwards at least as fast as the bowl, which means that for every
$\theta,\theta'\in\Theta_+$
\[
f(\theta')\ge f(\theta)+\nabla f(\theta)'(\theta'-\theta)
+\tfrac12(\theta'-\theta)'\Lambda(\theta'-\theta).
\]
For $\gamma>1$, utility is bounded above and tends to minus infinity as wealth
approaches zero. The quadratic penalty makes $f$ coercive, and the utility barrier
keeps each finite sublevel set away from the boundary of $\Theta_+$. Nonempty
finite sublevel sets are therefore compact, so a minimum is attained. Strong
convexity makes it unique. The same bounded-utility argument makes the posterior
normalising constant finite; the nonempty open domain gives it positive mass.
Next, the right-hand side of the inequality is a quadratic in
$\theta'$ whose smallest value, reached at
$\theta'=\theta-\Lambda^{-1}\nabla f(\theta)$, is
$f(\theta)-\tfrac12\nabla f(\theta)'\Lambda^{-1}\nabla f(\theta)$. Since $f$ lies
above this quadratic everywhere, its minimum cannot be lower, which is the stated
bound. The turnover penalty adds $\tfrac12\theta'M\theta$ with $M$ positive
semidefinite, which only increases the curvature, so the same argument applies
with $\Lambda+M$ in place of $\Lambda$. \hfill$\square$

\medskip
\noindent The bound uses only the gradient at the current estimate, which a
Newton solver computes in any case. We accept an estimate only when the bound is
sufficiently small, including provisional and bootstrap fits; a failed solve is
rejected (Appendix~\ref{sec:estimation}). This bounds objective error, rather
than coefficient distance directly. Strong convexity also gives
$\|\theta-\theta^*\|_\Lambda\le\sqrt{2\{f(\theta)-f(\theta^*)\}}$.

\section{Estimation Details}
\label{sec:estimation}
This appendix describes the computations behind the algorithms in the paper.

\subsection{Predictors}
\label{app:signals}
Let $R_{k,t}$ be the return on series $k$ in month $t$. The factor-specific
predictors are computed from these returns. Time-series momentum is the average
of $R_{k,t-11},\ldots,R_{k,t}$. Cross-sectional momentum ranks these averages
across the six series and maps the ranks to $[-1,1]$. With
$v_{k,t}=\sum_{s\le t}\log(1+R_{k,s})$, the valuation predictor is
$-(v_{k,t}-60^{-1}\sum_{j=0}^{59}v_{k,t-j})$. Reversal is
$-\{R_{k,t-1}-36^{-1}\sum_{j=0}^{35}R_{k,t-j}\}$, and volatility is minus the
standard deviation of the last twelve returns. At startup, the implementation
leaves the first complete look-back observation missing and starts each family
one observation later. Expanding standardisation leaves its first twelve rows
at zero. These startup conventions are retained throughout the replication.
Every predictor is standardised with the mean and
standard deviation of its history up to the current date, winsorised at three
standard deviations, and set to zero when missing. The vector $z_t$ appends a
constant to the \NSignals{} standardised predictors, so $L=\NSignalsPolicy$ and
the policy has $5\times\NSignalsPolicy=\NCoefficientsPolicy$ coefficients.

\subsection{Criterion}
Define $G_s\in\mathbb R^{(K-1)\times L}$ by $(G_s)_{k\ell}=z_{s,\ell}F_{k,s+1}$.
In the matrix calculations below, identify $G_s$ and $\theta$ with their
vectorisations, and let $G$ stack the vectorised score rows.
Wealth in \eqref{eq:criterion} is then
\begin{equation}\label{eq:affinewealth}
q_s(\theta)=1+r_{m,s+1}+\langle\theta,G_s\rangle ,
\end{equation}
which is linear in $\theta$. Because $U$ is concave and increasing, the utility
term of \eqref{eq:criterion} is convex, and the prior adds the strongly convex
term $\tfrac12\theta'\Lambda\theta$. Linear trading costs are omitted from this
criterion, and position and gross limits enter deployment rather than $q_s$.
Inserting bounding inside utility would define a different, generally nonconvex
estimating problem. Our convexity and mode apply to this affine
training surrogate, not to constrained realised utility.

\subsection{Computing the mode}
At the minimum of \eqref{eq:criterion},
$\Lambda\theta=\omega\sum_sU'\{q_s(\theta)\}G_s$. The mode is therefore a
combination $\theta=\Lambda^{-1}\sum_s\alpha_sG_s$ of the \TimingWindow{} score
vectors, and we solve for the \TimingWindow{} weights $\alpha$ rather than the
\NCoefficientsPolicy{} coefficients. The relevant kernel factorises,
\begin{equation}\label{eq:kernel}
K_{ss'}=\langle G_s,\Lambda^{-1}G_{s'}\rangle
=(F_{s+1}'F_{s'+1})\,(z_s'\Lambda_L^{-1}z_{s'}),
\end{equation}
with $\Lambda_L=\operatorname{diag}(\lambda_\ell)$, and wealth becomes
$q=1+r_m+K\alpha$. The first-order condition
$\rho(\alpha)=\alpha-\omega q(\alpha)^{-\gamma}=0$ is solved by Newton's method
from $\alpha=0$, the market, with a backtracking line search on
\eqref{eq:criterion}.

Strong convexity gives the bound of Lemma~\ref{lem:certified}, which in
this parametrisation is
\begin{equation}\label{eq:certificate}
f(\theta)-\min_{\theta'}f(\theta')\;\le\;
\tfrac12\,\nabla f(\theta)'\Lambda^{-1}\nabla f(\theta)
=\tfrac12\,\rho(\alpha)'K\rho(\alpha).
\end{equation}
We accept final, provisional and bootstrap estimates only when this bound is sufficiently small. 

\subsection{Learning rate}
For a regular unpenalised M-estimator with nonsingular population average
curvature $H$ and long-run score covariance $V$, matching covariance traces
motivates
$\omega=\operatorname{tr}(H^{-1})/\operatorname{tr}(H^{-1}VH^{-1})$. We use the
empirical diagonals of $H$ and $V$ as a heuristic, the latter estimated with
Bartlett weights,
\begin{equation}\label{eq:diagomega}
\hat\omega_t=
\frac{\sum_{j\in\mathcal I_t}H_{jj}^{-1}}
     {\sum_{j\in\mathcal I_t}V_{jj}H_{jj}^{-2}},
\qquad \mathcal I_t=\{j:H_{jj}>\xi\}.
\end{equation}
where $\xi$ is set to a very small number. Both diagonals are outer products over factors and predictors, so neither matrix
is formed. We compute the mode at a provisional rate, evaluate $\hat\omega_t$
there, and recompute the mode at $\hat\omega_t$. The median rate is
\TimingOmega. This is a one-pass rule, not a fixed-point solution at the final
mode. The empirical loss Hessian has rank at most $T$ in $(K-1)L$ dimensions,
so its full inverse does not exist here. Replacing matrices by diagonals neither
recovers that inverse nor establishes covariance matching for the penalised
estimator. The rule changes under rotations of the predictor coordinates even
when the likelihood and prior are unchanged. We retain the named, standardised
predictors as its coordinates and treat the resulting posterior uncertainty as
conditional on this choice. The covariance-trace motivation is distinct from
the inverse-covariance trace criterion of \citet{Lyddon2019}; neither supplies
a high-dimensional coverage theorem for this implementation. 

\subsection{Averaged weights}
\label{app:average}
Under the Gaussian approximation $\mathcal N(\hat\theta_t,C_t)$ to the posterior,
the unconstrained weight $X=\theta_k'z_t$ is normal with mean
$\mu_k=\hat\theta_k'z_t$ and variance $\sigma_k^2=z_t'C_{kk,t}z_t$, and its
constrained average is
\begin{equation}\label{eq:closedform}
\mathbb E[\Pi_c(X)]=-c\Phi(a_k)+\mu_k\{\Phi(b_k)-\Phi(a_k)\}
+\sigma_k\{\varphi(a_k)-\varphi(b_k)\}+c\{1-\Phi(b_k)\},
\end{equation}
with $a_k=(-c-\mu_k)/\sigma_k$ and $b_k=(c-\mu_k)/\sigma_k$. By the Woodbury
identity,
\begin{equation}\label{eq:woodbury}
\sigma^2_k=z_t'\Lambda_L^{-1}z_t-h_k'\bigl(D^{-1}+K\bigr)^{-1}h_k,
\qquad h_{k,s}=F_{k,s+1}\,(z_s'\Lambda_L^{-1}z_t),
\end{equation}
where $D=\omega\gamma\operatorname{diag}(q^{-\gamma-1})$, so only a
\TimingWindow-dimensional system is solved. The factor weights are then scaled
down in proportion whenever their absolute sum exceeds $g=\TimingLeverage$.

\subsection{Exact posterior}
\label{app:exact}
The criterion depends on $\theta$ only through $u=G\theta$, the vector of
\TimingWindow{} wealth increments. Under the prior, $u\sim\mathcal N(0,K)$,
possibly degenerate on $\operatorname{range}(K)$, and
$u$ and $X_k=\theta_k'z_t$ are jointly normal with $\operatorname{Cov}(u,X_k)=h_k$.
The posterior of $u$ has density proportional to
$\mathcal N(u;0,K)\exp\{\omega\sum_sU(1+r_{m,s+1}+u_s)\}$ relative to that
support, restricted to positive wealth, and is log-concave.
Given $u$, the criterion carries no information about $\theta$, so $X_k$ is
normal with mean $h_k'K^{+}u$ and variance
$z_t'\Lambda_L^{-1}z_t-h_k'K^{+}h_k$, where $K^+$ is the Moore--Penrose
pseudoinverse. We retain nonzero eigenvalues;
integration has dimension $\operatorname{rank}(K)\le T$. The first seven windows
have ranks 113 through 119, and the smallest retained rank is \RevExactMinRank.
The exact-target averaged weight is the
expectation over $u$ of \eqref{eq:closedform} evaluated at this conditional mean
and variance. We estimate the expectation over $u$ by Hamiltonian Monte Carlo,
preconditioned with the curvature at the mode (32 chains, 300 burn-in and 700
retained iterations, 12 steps of size 0.3 in whitened coordinates).
The replication uses twelve date shards with seeds 1 through 12. The median acceptance rate is
\RevExactAccept\%. We retain chain means, variances, classical split
$\widehat R$, and pre-gross-limit weight MCSEs using a Bartlett estimate with
100 lags at every date and cap. At the reported cap, the maximum split
$\widehat R$ is \RevExactRhatMax; median and maximum weight MCSEs are
\RevExactMcseMedian{} and \RevExactMcseMax. These finite-run diagnostics do not
prove convergence; $\widehat R$ here is not rank normalised. The exact-target
and approximate standard deviations of the unconstrained positions $X_k$
have a median ratio of \RevExactSdRatio. Their exact-target means differ from
the mode-implied positions by a median of \RevExactMeanShift{} Laplace posterior
standard deviations.

\begin{table}[ht]\centering
\caption{Gaussian approximation and exact posterior}\label{tab:exact}
\small
\begin{tabular}{lrrrrr}
\toprule
& Timing & $t$ & Gross SR & Net SR & Exposure \\
\midrule
Gaussian approximation (reported) & 2.08 & 3.00 & 1.242 & 1.043 & 1.31 \\
Exact target (Monte Carlo) & 2.12 & 2.95 & 1.242 & 1.035 & 1.32 \\
\bottomrule
\end{tabular}

\begin{minipage}{\textwidth}\footnotesize
\textit{Notes:} BPPP with the averaged weights computed under the Gaussian
approximation used throughout the paper and Monte Carlo estimates under the
exact target, at all
\RevExactDates{} dates. Timing in annual percentage points with its Newey--West
$t$-statistic.
\end{minipage}
\end{table}

\subsection{Turnover penalty}
Let $\Delta Z$ collect the monthly changes in the predictors within the
estimation window. The penalty in \eqref{eq:costcriterion} adds
$\omega\kappa\phi_k\Delta Z'\Delta Z$ to the prior precision of factor $k$'s
coefficients. The precision becomes block diagonal with blocks
$A_k=\Lambda_L+\omega\kappa\phi_k\Delta Z'\Delta Z$, and $\Lambda_L^{-1}$ is
replaced by $A_k^{-1}$ in \eqref{eq:kernel} and \eqref{eq:woodbury}, so that
$K=\sum_k(F_kF_k')\circ(ZA_k^{-1}Z')$. The Newton iterations and the bound
\eqref{eq:certificate} are otherwise unchanged, and every turnover-penalised
estimate is accepted only when its bound is sufficiently small. The learning rate is calibrated
without the penalty and held fixed when it is added, so that only $\kappa$ changes.

\subsection{Timing, costs and statistics}
Predictors observed at the end of month $t-1$ determine the weights held during
month $t$, and the estimate used at $t$ relies only on data up to $t-1$. The
estimation window contains the preceding \TimingWindow{} months. Performance is
evaluated from \OOSStart{} to \OOSEnd, \NOOS{} months.

Net returns deduct proportional trading costs of 60, 25, 30, 30 and 70 basis
points for changes in the SMB, HML, CMA, RMW and UMD weights, and annual holding
costs of 90, 40, 45, 45 and 80 basis points on the absolute weights, charged
monthly. In the equation, $h_k$ denotes the quoted annual decimal rate divided
by 12, whereas $\phi_k$ is a one-way trading cost of
\begin{equation}\label{eq:netcriterion}
r^{\,n}_{t+1}=r_{m,t+1}+w_t'F_{t+1}
-\sum_k \phi_k|w_{k,t}-w_{k,t-1}|
-\sum_k h_k|w_{k,t}|.
\end{equation}
The market weight is fixed and incurs no cost, and the initial purchase is not
charged. The costs are applied at the factor level; rebalancing of the
underlying stocks, price impact, borrowing constraints and netting across
factors are not modelled.

The Sharpe ratio is $\sqrt{12}\,\bar r_p/s(r_p)$ and the certainty equivalent is
$12[\{T^{-1}\sum_t(1+r_{p,t})^{1-\gamma}\}^{1/(1-\gamma)}-1]$. Turnover is twelve
times the average monthly sum of absolute changes in the factor weights. Exposure
is the average sum of absolute factor weights. The timing component is gross of
costs, uses divisor $T$ for covariance and a Newey--West estimate of its standard
error. Return and timing HAC inference uses Bartlett weights with bandwidth
$\lfloor4(T/100)^{2/9}\rfloor$, equal to five at the full-sample length. Paired
Sharpe differences use the delta-method HAC construction in Section~3.1 of
\citet{LedoitWolf2008}, rather than their studentised time-series bootstrap.

\clearpage
\begin{landscape}
\section{Additional Tables}
\label{sec:apptables}
\begin{table}[!h]\centering
\caption{Predictors that time each factor}\label{tab:signalsbyfactor}
\setlength{\tabcolsep}{2.5pt}\scriptsize
\resizebox{\linewidth}{!}{\begin{tabular}{@{}lrrr@{\hspace{10pt}}lrrr@{\hspace{10pt}}lrrr@{\hspace{10pt}}lrrr@{\hspace{10pt}}lrrr@{}}
\toprule
\multicolumn{4}{c}{SMB} & \multicolumn{4}{c}{HML} & \multicolumn{4}{c}{CMA} & \multicolumn{4}{c}{RMW} & \multicolumn{4}{c}{UMD} \\
\cmidrule(lr){1-4}\cmidrule(lr){5-8}\cmidrule(lr){9-12}\cmidrule(lr){13-16}\cmidrule(lr){17-20}
Predictor & Timing & $\beta$ & $\rho$ & Predictor & Timing & $\beta$ & $\rho$ & Predictor & Timing & $\beta$ & $\rho$ & Predictor & Timing & $\beta$ & $\rho$ & Predictor & Timing & $\beta$ & $\rho$ \\
\midrule
Leverage & $-$0.303 & 0.079 & $-$0.091 & Leverage & 0.664 & 0.115 & 0.133 & SP & 0.148 & 0.050 & 0.115 & AM & $-$0.152 & $-$0.066 & 0.073 & zerotrade6M & 0.238 & 0.044 & 0.099 \\
VolMkt & 0.251 & $-$0.039 & $-$0.195 & AM & $-$0.515 & $-$0.097 & 0.128 & Vol\_4 & $-$0.116 & 0.026 & $-$0.176 & SP & 0.133 & 0.038 & 0.119 & NetEquityFinance & 0.233 & $-$0.048 & $-$0.084 \\
zerotrade6M & 0.198 & $-$0.039 & $-$0.136 & CSMom\_5 & 0.489 & $-$0.119 & $-$0.105 & AM & $-$0.111 & $-$0.047 & 0.085 & XFIN & $-$0.085 & $-$0.026 & 0.105 & Price & $-$0.165 & $-$0.050 & 0.066 \\
CSMom\_0 & 0.181 & $-$0.069 & $-$0.074 & SP & 0.387 & 0.070 & 0.142 & Leverage & 0.108 & 0.039 & 0.096 & Val\_1 & $-$0.083 & 0.052 & $-$0.073 & Val\_2 & 0.162 & $-$0.054 & $-$0.048 \\
ShareIss5Y & 0.180 & $-$0.033 & $-$0.124 & CSMom\_1 & $-$0.361 & $-$0.138 & 0.080 & CSMom\_0 & 0.081 & $-$0.037 & $-$0.090 & Vol\_2 & $-$0.082 & 0.027 & $-$0.090 & zerotrade12M & $-$0.142 & $-$0.029 & 0.091 \\
Beta & $-$0.170 & $-$0.032 & 0.143 & DelCOA & $-$0.269 & $-$0.049 & 0.160 & dNoa & 0.077 & 0.023 & 0.145 & Tax & 0.076 & 0.018 & 0.145 & ChEQ & 0.141 & 0.046 & 0.060 \\
std\_turn & 0.160 & $-$0.033 & $-$0.119 & CSMom\_4 & 0.261 & $-$0.094 & $-$0.072 & Val\_5 & 0.076 & $-$0.026 & $-$0.096 & Vol\_3 & $-$0.074 & 0.028 & $-$0.086 & Vol\_0 & $-$0.137 & $-$0.019 & 0.129 \\
XFIN & 0.159 & $-$0.040 & $-$0.097 & Vol\_2 & $-$0.259 & 0.039 & $-$0.150 & HerfAsset & 0.072 & $-$0.023 & $-$0.109 & HerfAsset & 0.071 & $-$0.016 & $-$0.136 & Leverage & $-$0.124 & 0.048 & $-$0.043 \\
Vol\_5 & 0.145 & $-$0.035 & $-$0.103 & TSMom\_2 & $-$0.250 & $-$0.048 & 0.120 & TSMom\_2 & $-$0.071 & $-$0.019 & 0.131 & HerfBE & $-$0.070 & 0.017 & $-$0.126 & CSMom\_0 & 0.116 & 0.024 & 0.093 \\
BetaFP & 0.135 & 0.024 & 0.149 & dNoa & 0.212 & 0.064 & 0.094 & CashProd & 0.069 & 0.021 & 0.129 & Spinoff & $-$0.063 & 0.019 & $-$0.121 & BrandInvest & 0.116 & 0.018 & 0.146 \\
\midrule
Total timing & 0.625 & & & Total timing & 0.758 & & & Total timing & 0.078 & & & Total timing & 0.046 & & & Total timing & 0.573 & & \\
\bottomrule
\end{tabular}
}

\par\vspace{0.75em}
\noindent\begin{minipage}{\linewidth}\scriptsize
\textit{Notes:} Each factor column group lists the ten predictors that contribute most, in
absolute value, to the timing component of one factor under BPPP. With
$w_k-\bar w_k=\sum_\ell\beta_\ell(z_\ell-\bar z_\ell)+e$, predictor $\ell$
contributes $\beta_\ell\operatorname{Cov}(z_\ell,F_k)$, in annual percentage
points; $\rho$ is the correlation of the predictor with the factor's return in the
following month. Anomaly acronyms follow \citet{ChenZimmermann2022}. The
factor-specific predictors are labelled by family and series, with series
numbered 0 (Mkt), 1 (SMB), 2 (HML), 3 (CMA), 4 (RMW) and 5 (UMD). Total timing
is the factor's full timing component; the ten contributions shown do not sum to
it. The attribution is in sample, and individual contributions are not separately
identified when predictors are correlated.
\end{minipage}
\end{table}
\end{landscape}

\begin{table}[p]\centering
\caption{Positions and turnover}\label{tab:diagnostics}
\small
\begin{tabular}{lrrrrr}
\toprule
& Exposure & At position & At gross & Turnover & Gross \\
& & limit (\%) & limit (\%) & & SR \\
\midrule
BPPP & 1.31 & 0.0 & 52.6 & 5.01 & 1.242 \\
RPPP & 1.49 & 0.6 & 96.4 & 5.49 & 1.211 \\
Constant allocation & 1.50 & 26.4 & 100.0 & 0.15 & 1.193 \\
Lasso & 1.40 & 4.1 & 80.2 & 1.88 & 1.094 \\
Conditional mean--variance & 1.50 & 0.2 & 99.2 & 2.00 & 1.160 \\
Minimum variance & 1.48 & 2.4 & 86.6 & 0.68 & 1.106 \\
Inverse variance & 1.00 & 0.0 & 0.0 & 0.27 & 1.002 \\
Reward to risk & 0.96 & 5.0 & 0.0 & 1.20 & 0.986 \\
Scaled minimum variance & 1.10 & 1.2 & 43.1 & 0.73 & 1.018 \\
Market & 0.00 & 0.0 & 0.0 & 0.00 & 0.739 \\
\bottomrule
\end{tabular}

\begin{minipage}{\textwidth}\footnotesize
\textit{Notes:} All rules face a position limit of \TimingCap{} and a
gross-exposure limit of \TimingLeverage. Exposure is the average sum of absolute
factor weights. At position limit is the percentage of factor-months in which a
weight equals $\pm\TimingCap$. At gross limit is the percentage of months in
which the sum of absolute final factor weights equals $\TimingLeverage$.
This boundary statistic applies to all rules, including the constant allocation
that imposes the limits jointly during optimisation.
Turnover is the annual sum of absolute changes in factor weights.
\end{minipage}
\end{table}

\end{document}